\documentclass[12pt]{article}
\usepackage{amsfonts}
\usepackage{amsmath}
\usepackage{makeidx}
\usepackage{amssymb}
\usepackage{graphicx}
\usepackage[T1]{fontenc}

\newtheorem{theorem}{Theorem}

\newtheorem{corollary}[theorem]{Corollary}

\newtheorem{adefinition}[theorem]{Definition}

\newtheorem{lemma}[theorem]{Lemma}

\newtheorem{proposition}[theorem]{Proposition}
\newtheorem{aremark}[theorem]{Remark}
\newenvironment{remark}{\begin{aremark}\rm}{\end{aremark}}

\newenvironment{proof}[1][Proof]{\noindent\textbf{#1.} }{\ \rule{0.5em}{0.5em}}
\numberwithin{equation}{section} \numberwithin{theorem}{section}

\begin{document}

\title{On Random Matrices Arising \\
in Deep Neural Networks: General I.I.D. Case }

\author{L. Pastur and V. Slavin \\
%EndAName
{\small B.Verkin Institute for Low Temperature Physics and Engineering}\\
{\small Kharkiv, Ukraine}}

\date{}

\maketitle

%\begin{history}
%\received{(Day Month Year)} \revised{(Day Month Year)}
%\end{history}

\begin{abstract}
We study the eigenvalue distribution of random matrices pertinent
to the analysis of deep neural networks. The matrices resemble the
product of the sample covariance matrices, however, an important
difference is that the analog of the population covariance matrix
is now a function of random data matrices (synaptic weight
matrices in the deep neural network terminology). The problem has
been treated in recent work \cite{Pe-Co:18} by using the
techniques of free probability theory. Since, however, free
probability theory deals with population covariance matrices which
are independent of the data matrices, its applicability in this
case has to be justified. The justification has been given in
\cite{Pa:20} for Gaussian data matrices with independent entries,
a standard analytical model of free probability, by using a
version of the techniques of random matrix theory. In this paper
we use another version of the techniques to
extend the results of \cite{Pa:20} to the case where the
entries of the data matrices are just independent identically
distributed random variables with zero mean and finite fourth
moment. This, in particular, justifies the mean field approximation
in the infinite width  limit for the deep untrained neural networks and
the property of the macroscopic universality of random matrix theory
in this case.
\end{abstract}

%\keywords{random matrices; deep neural networks.}

%\ccode{Mathematics Subject Classification 2000: 15B52; 92B20}

%%%%%%%%%%%%%%%%%%%%%%%%%%%%%%%%%%%%%%%%%%%%%%%%%%%%%%%%%%%%%%%%%%%%%%%%%%%%%%%%

\section{Introduction}

Deep learning, a powerful computational technique based on deep artificial neural
networks (DNN) of various architecture, proved to be an efficient tool in a
wide variety of problems involving large data sets, see, e.g. \cite%
{Bu:17,Ca-Ch:18,Go-Co:16,Le-Co:15,Ro-Co:22,Sh-Ma:19}. A general scheme for
the so-called feed-forward, fully connected neural networks with $L$ layers
of width $n_{l}$ for the $l$th layer
%and a nonlinearity $\varphi $
is as follows.

Let%
\begin{equation}
x^{0}=\{x_{j_{0}}^{0}\}_{j_{0}=1}^{n_{0}}\in \mathbb{R}^{n_{0}}  \label{x0}
\end{equation}%
be the  \emph{input} to the network and $x^{L}=\{x_{j_{L}}^{L}\}_{j_{L}=1}^{n_{L}}%
\in \mathbb{R}^{n_{L}}$ be the \emph{output}. Their components are known as the \emph{neurons} (the terminology here and below is inspired by that in biological neural networks).
The components of the \emph{activations} $x^{l}=\{x_{j_{l}}^{l}\}_{j_{l}=1}^{n_{l}},\; l=1,\dots, L-1$ and the components of
the post-affine transformations $y^{l}=\{y_{j_{l}}^{l}\}_{j_{l}=1}^{n_{l}}$
in the $l$th layer
%known as the \emph{neurons} and
are related via an affine transformation
\begin{equation}
y^{l}=W^{l}x^{l-1}+b^{l},\;x_{j_{l}}^{l}=\varphi
(y_{j_{l}}^{l}),\;j_{l}=1,\dots,n_{l},\;l=1,\dots,L,  \label{rec}
\end{equation}%
where
\begin{equation}
W^{l}=\{W_{j_{l}j_{l-1}}^{l}\}_{j_{l},j_{l-1}=1}^{n_{l},n_{l-1}},\;l=1,\dots,L
\label{wl}
\end{equation}%
are $n_{l}\times n_{l-1}$ \emph{weigh}t matrices,%
\begin{equation}
b^{l}=\{b_{j_{l}}^{l}\}_{j_{l}=1}^{n_{l}},\;l=1, \dots,L  \label{bl}
\end{equation}%
are $n_{l}$-component \emph{bias} vectors and $\varphi :\mathbb{R}\rightarrow
\mathbb{R}_+$ is the component-wise nonlinearity known as the \emph{activation function}.
%We will refer to these as the post- and pre-activations. (We let x
%0 i xi for the input, dropping the Arabic numeral superscript,
%and instead use a Greek superscript x to denote a particular
%input).
%Weight and bias parameters for the $l$th layer form the $n^{l}
%\times n^{l-1}$ rectangular random  matrix
%$W^l=\{W^l_{jk}\}_{j,k=1}^{n_l,n_{l-1}}$ and the vector
%$b_l=\{b^l_j\}_{j=1}^{N^l}$, and since this is an untrained neural
It is usually monotone and piece-wise differentiable
(S-shaped or sigmoid), e.g. $tanh$, $tan^{-1}$, $(1+e^{-x})^{-1}$ and HardTanh (see (\ref{hata})).
A widely used and fast calculated activation function is the \emph{rectified linear unit} (ReLU)
$x_+ := \max\{0,x\}$.

An important ingredient of the deep learning is the training procedure. It
modifies the parameters (weight matrices and biases) on the every step of
the iteration to reduce the
misfit between the input and the output data of the layer by using certain
optimization procedures, usually the stochastic gradient descend (SGD).
Being multiply repeated in the DNN, the procedure provides the desired final output
as well as certain final parameters of the DNN in question.
%, see e.g. \cite{Ma-Ma:17}.

In the DNN practice the weights and the biases are randomly initialized and the SGD
also includes a certain randomization. Moreover, the modern theory deals also with
untrained and even random parameters of the DNN architecture, see  \cite{Pe-Co:18,Ba-Co:20,Ro-Co:22,Ca-Sc:20,Gi-Co:16,Li-Qi:00,Ma-Co:16,Pe-Ba:17,Po-Co:16,Sc-Wa:17,Sc-Co:17,
Ta-Co:18,Wa-Co:18,Ya:20}.
It is often assumed in these and other works that the weight matrices and biases are independent and identically distributed (i.i.d.) in $l$ and have i.i.d. Gaussian entries and components.

Following this trend in the DNN studies and taking into account that
quite common initialization schemes in deep learning do not use
Gaussians \cite{Gl-Be:10} on
one hand and recalling the independence of
various results of random matrix theory on the concrete distribution of the parameters (macroscopic universality) \cite{Pa-Sh:11,Ta-Vu:10} on the other hand,
we consider in this paper a general i.i.d. case where:

 (i) the bias vectors $b^{l},\;l=1,2,\dots,L$ are i.i.d. in $l$
and for every $l$ their components $\{b_{j_{l}}^{l}\}_{j_{l}=1}^{n_{l}}$ are
i.i.d. random variables such that
\begin{equation}
\mathbf{E}\{b_{j_{l}}^{l}\}=0,\;\mathbf{E}\{(b_{j_{l}}^{l})^{2}\}=\sigma
_{b}^{2},  \label{bga}
\end{equation}

  (ii) the weight matrices $W^{l},\;l=1,2 \dots,L$ are also i.i.d in $l$
and
\begin{align}
W^{l}&
=n_{l-1}^{-1/2}X^{l}=n_{l-1}^{-1/2}\{X_{j_{l}j_{l-1}}^{l}%
\}_{j_{l},j_{l-1}=1}^{n_{l},n_{l-1}},  \notag \\
\mathbf{E}\{X_{j_{l}j_{l-1}}^{l}\}& =0,\;\mathbf{E}%
\{(X_{j_{l}j_{l_{1}-1}}^{l})^{2}\}=w^{2},\;\mathbf{E}%
\{(X_{j_{l}j_{l_{1}-1}}^{l})^{4}\}=m_{4}<\infty,  \label{wga}
\end{align}%
where for every $l$ the entries $\{X_{j_{l}j_{l-1}}^{l}%
\}_{j_{l},j_{l-1}=1}^{n_{l},n_{l-1}}$ of $X^{l}$ are i.i.d. random variables.

Note that the quite common initialization schemes in deep learning do not use
Gaussians, see, e.g. \cite{Gl-Be:10}.

We will view $n_{l}\times n_{l-1}$ matrices $X^{l}$ as the upper left
rectangular blocks of the semi-infinite random matrix%
\begin{equation}
\{X_{j_{l}j_{l-1}}^{l}\}_{j_{l},j_{l-1}=1}^{\infty ,\infty }  \label{xinf}
\end{equation}%
whose i.i.d. entries satisfy (\ref{wga}).

Likewise, for every $l$ we will view $b^{l}$ in (\ref{bl}) as the first $n_{l}$
components of the semi-infinite vector
\begin{equation}
\{b_{j_{l}}^{l}\}_{j_{l}=1}^{\infty }  \label{binf}
\end{equation}%
whose i.i.d. components satisfy (\ref{bga}).

As a result of this form of weights and biases of the $l$th layer they are
for all $n_{l}=1,2, \dots$ defined on the same infinite-dimensional product
probability space $\Omega ^{l}$ generated by (\ref{xinf}) -- (\ref{binf}).
Let also
\begin{equation}  \label{oml}
\Omega _{l}=\Omega ^{l}\times \Omega ^{l-1}\times \dots \times \Omega
^{1},\;l=1,\dots,L
\end{equation}%
be the infinite-dimensional probability space on which the recurrence (\ref%
{rec}) is defined for a given $L$ (the number of layers).

This procedure of enlarging the probability space is standard in probability theory, see e.g. [37]. In our case the procedure allows us to formulate our results on
the large size asymptotic behavior of the eigenvalue distribution of matrices (\ref{JJM})
as those valid with probability 1 in $\Omega_L$, i.e., for any "typical" realization of random
parameters (for an analogous approach in random matrix theory see, e.g. \cite{Pa-Sh:11})."

%Note that matrices $W^{l}(W^{l})^{T}$ of (\ref{wga}) are known in statistics
%as the Wishart matrices \cite{Pa-Sh:11,Mu:05} and their limting NCM is.

A key and quite non-trivial step in training deep networks is to move the weights $W^l$,
hence, the activations $x^{l}$ in each layer $l$ so as to move the output $x^L$ in the final
layer in a desired direction. To this end the analysis of the back propagation of
errors at the output of layer $l$, determining how we have to change $x^l$ to move $x^L$,
is useful. A corresponding tool is the Jacobian $\frac{\partial x^L}{\partial x^0}$ showing
how an error $\varepsilon$, or desired direction of motion in the output $x^L$, back-propagates
to a desired change in the input $\Delta (x^0)^T = \varepsilon^T J$, see more in
\cite{Ba-Co:20,Sc-Co:17}.

The above makes the Jacobian an important quantity of the field. We have according to (\ref{x0}) -- (\ref{bl})
\begin{equation}
J_{\mathbf{n}_{L}}^{L}:=\left\{ \frac{\partial x_{j_{L}}^{L}}{\partial
x_{j_{0}}^{0}}\right\}
_{j_{0},j_{L}=1}^{n_{0},n_{L}}= D^{L}W^{L} \cdots D^{1}W^{1},\;%
\mathbf{n}_{L}=(n_{1},\dots,n_{L}),  \label{jac}
\end{equation}%
an $n_{L}\times n_{0}$ random matrix, where $\{W^l \}_{l=1}^L$ are given by (\ref{wga})
and\begin{equation}
D^{l}=\{D_{j_{l}}^{l}\delta
_{j_{l}k_{l}}\}_{j_{l},k_{l}=1}^{n_{l}},\;D_{j_{l}}^{l}=\varphi ^{\prime }%
\Big(n_{l-1}^{-1/2}%
\sum_{j_{l-1}=1}^{n_{l-1}}X_{j_{l}j_{l-1}}^{l}x_{j_{l-1}}^{l-1}+b_{j_{l}}^{l}%
\Big), \; l=1,\dots,L  \label{D}
\end{equation}%
are diagonal random matrices.

%The input-output Jacobian $J_{\mathbf{n}_{L}}^{L}$ is closely %related to the
%backpropagation operator mapping output errors to weight matrices at a given
%layer. %, in the sense that if the former is
Of particular interest is the spectrum of singular values of $J_{\mathbf{n}%
_{L}}^{L}$, i.e., the square roots of eigenvalues of the $n_{L}\times n_{L}$
positive definite matrix
\begin{equation}
M_{\mathbf{n}_{L}}^{L}=J_{\mathbf{n}_{L}}^{L}(J_{\mathbf{n}_{L}}^{L})^{T}
\label{JJM}
\end{equation}%
for networks with the above random weights and biases and for large $%
\{n_{l}\}_{l=1}^{L}$, i.e., for deep networks with wide layers but with a
fixed depth $L$, see \cite
{Pe-Co:18,Li-Qi:00,Ma-Co:16,Pe-Ba:17,Po-Co:16,Sc-Co:17,Ta-Co:18,Ya:20} for
motivations, settings and results. More precisely, we will study in this
paper the asymptotic regime %, where
%\begin{equation}
%$1\ll L\ll n_{l},\;l=1,2,...\;  \label{asi}
%\end{equation}%
determined by the simultaneous limits
\begin{equation}
n_{l}\rightarrow \infty ,\;n_{l-1}/n_{l} \to c_{l}\in (0,\infty
),\;_{{}}\;l=1,\dots,L  \label{asf1}
\end{equation}%
denoted below as
\begin{equation}
\lim_{\mathbf{n}_{L}\rightarrow \infty } \dots \label{limn}
\end{equation}%
%and followed by the limit
%\begin{equation}
%L\rightarrow \infty .  \label{asf2}
%\end{equation}%
The above limit can be viewed as an
implementation of the heuristic inequality $ L \ll n$, meaning that the DNN
in question are much more wide than they are deep. The simplest case where $L=1$ and $D_1=\mathbf{1}_n$ is known in statistics as the Wishart matrices \cite{Pa-Sh:11,Mu:05}
and in this case the limiting NCM (see (\ref{ncm}) for definition) is
\begin{equation}\label{nump}
\nu _{MP}(\lambda)=(2\pi)^{-1}\sqrt{(4-\lambda)/\lambda}, \; \lambda \in [1,4].
\end{equation}
Denote by $\{\lambda _{t}^{L}\}_{t=1}^{n_{L}}$ the eigenvalues of
the real symmetric random matrix (\ref{JJM}) %$M_{\mathbf{n}_{L}}^{L}$
and introduce its
Normalized Counting Measure (NCM)
\begin{equation}
\nu _{M_{\mathbf{n}_{L}}^{L}}=n_{L}^{-1}\sum_{t=1}^{N_{L}}\delta _{\lambda
_{t}^{L}}.  \label{ncm}
\end{equation}%
We will deal with the leading term of $\nu _{M_{\mathbf{n}_{L}}^{L}}$in the
asymptotic regime (\ref{asf1}), i.e., with the limit%
%which we denote $\lim_{\mathbf{n}_{L}\rightarrow
%\infty }$:
\begin{equation}
\nu _{M^{L}}=\lim_{\mathbf{n}_{L}\rightarrow \infty }\nu _{M_{\mathbf{n}%
_{L}}^{L}}.  \label{ids}
\end{equation}%
Note that since $\nu _{M_{\mathbf{n}_{L}}^{L}}$ is random, the meaning of
the limit has to be indicated.

The problem has been considered in \cite{Pe-Co:18} (see also \cite%
{Pa:20,Ba-Co:20,Li-Qi:00,Ta-Co:18,Ya:20}) in the case where all $b^{l}$ and $%
X^{l},\;l=1,2,\dots,L$ in (\ref{bga}) -- (\ref{wga}) are Gaussian and have the
same size $n$ and $n\times n$ respectively, i.e.,%
\begin{equation}
n=n_{0}=\dots=n_{L},\;c_{l}=1,\;l=1,\dots,L.  \label{nequ}
\end{equation}%%
%the components of $$ i.i.d. in $l$, thus $\sigma^2_{b^l} \; w_l$ and $m^l_4$,
%We will write in this case $n$ instead of $n_{l},\;l=0,...,L$.
%Assume for the sake of simplicity of subsequent formulas that $w^{l}=1$ in (%
%\ref{wga}). the random matrices $X^{l}$ of (\ref{wga}) and the random
%vectors $b^{l}$ of (\ref{bga}) are i.i.d. in $l$ and that .

In \cite{Pe-Co:18} compact formulas for the limit
\begin{equation}
\overline{\nu }_{M^{L}}=\lim_{n\rightarrow \infty }\overline{\nu }%
_{M_{n}^{L}},\;\overline{\nu }_{M_{n}^{L}}:=\mathbf{E}\{\nu _{M_{n}^{L}}\}
\label{mids}
\end{equation}%
and its Stieltjes transform
\begin{equation}
f_{M^{L}}(z)=\int_{0}^{\infty }\frac{\overline{\nu }_{M^{L}}(d\lambda )}{%
\lambda -z},\;\Im z\neq 0  \label{stm}
\end{equation}%
were presented. The formula for $\overline{\nu }_{M^{L}}$ is given in (\ref%
{nucon}) below. To write the formula for $f_{M^{L}}$ it is convenient to use
the moment generating function
\begin{equation}
m_{M^{L}}(z)=\sum_{k=1}^{\infty }m_{k}z^{k},\;m_{k}=\int_{0}^{\infty
}\lambda ^{k}\overline{\nu }_{M^{L}}(d\lambda ),  \label{mgen}
\end{equation}%
of $\overline{\nu }_{M^{L}}$ related to $f_{M^{L}}$ as
\begin{equation}
m_{M^{L}}(z)=-1-z^{-1}f_{M^{L}}(z^{-1}).  \label{stmg}
\end{equation}%
Let
\begin{equation}
K_{n}^{l}:=(D_{n}^{l})^{2}=\{(D_{j_{l}}^{l})^{2}\}_{j_{l}=1}^{n}  \label{kan}
\end{equation}%
be the square of the $n\times n$ random diagonal matrix (\ref{D}) with $%
n_{l}=n$
%, denoted $D_{n}^{l}$ to make explicit its dependence
%on $n$ of (\ref{nequ}),
and let $m_{K^{l}}$ be the moment generating function of the $%
n\rightarrow \infty $ limit $\overline{\nu }_{K^{l}}$ of the expectation of
the NCM of $K_{n}^{l}$. Then we have according to formulas (14) and (16) in
\cite{Pe-Co:18} in the case, where $\overline{\nu }_{K^{l}}$, hence $%
m_{K^{l}} $, do not depend on $l$ (see Remark \ref{r:penn} (i)),
\begin{align}
m_{M^{L}}(z)& =m_{K}(z^{1/L} \psi_{L}(m_{M^{L}}(z)),  \notag \\
\psi _{L}(m)& =m^{(L-1)/L}(1+m)^{1/L}.  \label{penfo}
\end{align}%
Hence, $f_{M^{L}}$ of (\ref{stm}) satisfies a certain functional equation,
the standard situation in random matrix theory and its applications, see
\cite{Pa-Sh:11} for general results and \cite{Go-Co:15,Mu:02} for results on
the products of random matrices. Note that our notation is different from
that of \cite{Pe-Co:18}: our $f_{M^{L}}(z)$ of (\ref{stm}) is $-G_{X}(z)$ of
(7) in \cite{Pe-Co:18} and our $m_{M^{L}}(z)$ of (\ref{mgen}) is $M^{-1}_{X}(1/z)$
of (9) in \cite{Pe-Co:18}.

The derivation of (\ref{penfo}) and the corresponding formulas (see (\ref%
{nucon}) -- (\ref{nucone}) below) for the limiting mean NCM\ $\overline{\nu }%
_{M^{L}}$ in \cite{Pe-Co:18} are based on the claimed in this paper asymptotic
freeness of diagonal matrices $D_{n_{l}}^{l}=\{D_{j_{l}}^{l}%
\}_{j_{l}=1}^{n_{l}},\;l=1,\dots,L$ of (\ref{D}) %with $n=n_{l}$
and Gaussian matrices $X_{n_{l}}^{l},\;l=1,\dots,L$ of (\ref{wl}) -- (\ref%
{wga}) %with $n_{l-1}=n_{l}=n$
(see, e.g. \cite{Mi-Sp:17} for the definitions and properties of
asymptotic freeness). This leads directly to (\ref{penfo}) in view of the
multiplicative property of the moment generating functions (\ref{mgen}) and
the so-called $S$-transforms of the  mean limiting NCM $\overline{\nu }_{K^{l}}$
of $K_{n_{l}}^{l}$ and the  mean limiting NCM $\overline{\nu}_{MP}$ (see (\ref{nump}))
%\begin{equation}\label{nump}
%\nu _{MP}=4\pi^{-1}\sqrt{(4-\lambda)/\lambda}
%\end{equation}
of $n^{-1}X_{n_{l}}^{l}(X_{n_{l}}^{l})^{T}$ in the regime (\ref{asf1}), see
Remark \ref{r:penn} (ii) and Corollary \ref{c:conv}.

There is, however, a delicate point in the argument of
\cite{Pe-Co:18}, since, to the best of our knowledge, the asymptotic
freeness has been established so far for the Gaussian random matrices $%
X_{n_{l}}^{l}$ of (\ref{wga}) and deterministic (more generally, random but $%
X_{n_{l}}^{l}$-independent) diagonal matrices, see e.g. \cite%
{Mi-Sp:17} and also \cite{Go-Co:15,Mu:02}.
On the other hand, the diagonal matrices $D_{n}^{l}$ in (\ref{D}) depend
explicitly on $(X_{n}^{l},b_{n}^{l})$ of (\ref{wl}) -- (\ref{bl}) and,
implicitly, via $x^{l-1}$, on the all preceding $(X_{n}^{l^{\prime
}},b_{n}^{l^{\prime }}),\;l^{\prime }=l-1,\dots,1$. Thus, the proof of
validity of (\ref{penfo}) requires an additional reasoning. \ It was given
in \cite{Pa:20} for the Gaussian weights and biases by using a version of
standard tools of random matrix theory (see \cite{Pa-Sh:11}, Chapter 7). Note that it was also proved in \cite{Pa:20} that the formula (\ref{ids}) is
valid not only in the mean (see (\ref{mids}) and \cite{Pe-Co:18}), but also
with probability 1 in $\Omega_{L}$ of (\ref{oml}) (recall that the measures
in the r.h.s. of (\ref{ids})\ are random) and that the corresponding limiting measure $\nu
_{M^{L}}$
%in the l.h.s. of (\ref{ids})
coincides with $\overline{\nu }%
_{M^{L}}$ of (\ref{mids}), i.e., $\nu _{M^{L}}$ is non-random (the selfaveraging property of the limiting NCM).

The basic ingredient of the proof in \cite{Pa:20} is the justification of
the replacement of the argument of $\varphi ^{\prime }$ in (\ref{D}), i.e.,
the post-affine $y_{j_{l}}^{l}$ of (\ref{rec}), by a Gaussian random variable which
is statistically independent of $\{X^{l^{\prime }}\}_{l^{\prime }=1}^l$, see
Lemma 3.5 of \cite{Pa:20}. This reduces the analysis of random matrices (\ref%
{JJM}) to that of random matrices with random but $X^l$-independent
analogs of diagonal matrices $D^{l}$, a well studied problem of random matrix theory,
see, e.g. \cite{Pa:20,Co-Ha:14}, and justifies the so-called infinite width mean-field
limit discussed in \cite{Ma-Co:16,Po-Co:16,Sc-Co:17,Ya:20}.

The goal of this paper is to show that the results presented in \cite%
{Pe-Co:18} and justified in \cite{Pa:20} (see also \cite{Ya:20}) for the
Gaussian weights and biases are valid for arbitrary random weights and
biases satisfying (\ref{bga}) -- (\ref{wga}). It is worth mentioning that
our initial intention was to carry out this extension just by using the
so-called interpolation trick of random matrix theory. The trick allows one
to extend a number of results of the theory valid for Gaussian matrices with
i.i.d. entries to those for matrices with i.i.d. entries possessing just
several finite moments, see \cite{Pa-Sh:11,Pa:00} (this is known as
the macroscopic, or global, universality). We have found, however, that in
our case the
corresponding proof is quite tedious and long. Thus, we apply another method
which dates back to \cite{Ma-Pa:67,Pa:72} and has been widely used and
extended afterwards \cite{Pa-Sh:11,Ak-Co:11,Ba-Si:00,Gi:12}. The method is quite
transparent and its application to matrices (\ref{JJM}) requires just minor
modifications of that used in random matrix theory where the analogs
of matrices $D^{l}$ of (\ref{D}) are either non-random or random but
independent of $X^{l}$ of (\ref{wga}).

The paper is organized as follows. In the next Section 2 we prove the validity
of (\ref{ids}) with probability 1 in $\Omega _{L}$ of (\ref{oml}), formula (%
\ref{penfo}) and the corresponding formula (\ref{nucon}) for $\nu _{M^{L}}=%
\overline{\nu }_{M^{L}}$ of \cite{Pe-Co:18}.  This is given in Theorem \ref{t:main} which proof is based on a natural inductive procedure allowing for the passage from the $l$th to the $(l+1)$th
layer and it is quite close to that of \cite{Pa:20}. This is because the procedure is almost independent on the probability law of the weight entries, provided that a formula relating the limiting (in the layer width) Stieltjes
transforms of the NCMs of two subsequent layers is known for these entries.
For the i.i.d. case of the present paper the formula is the same as that in \cite{Pe-Co:18,Pa:20} for the Gaussian case,
although its proof is quite different from that in \cite{Pa:20}. The formula
is given in Theorem \ref{t:ind}. Section 2 includes also certain numerical
results that illustrate and confirm our analytical results. Section 3 contains
the proof of Theorem \ref{t:ind} as well as necessary auxiliary results used in the proof.

Note that to make the paper self-consistent we present here certain facts
that have been already given in our previous paper \cite{Pa:20}. %Besides, in
%the subsequent paper \cite{Pa:21} we consider the case where the weight
%matrices are the Haar distributed orthogonal matrices, thereby justifying
%certain moments of the analysis of this case in \cite{Pe-Co:18}.
%, the proof of the formula requires an additional and, we
%believe, a non-trivial argument, presented%
%in Section 3 and justifying the coincidence of the
%limiting NCM of this case and that with random but independent of $X^{l}_{n_l}$
%certain matrices $\mathsf{D}_{n_l}^{l}$ (see Lemma \ref{l:inter}).

%It follows from the results of the section that the coincidence of the
%limiting eigenvalue distribution of matrices of two indicated cases is due
%to the form of dependence of $D_{n}^{l}$ on $(X_{n}^{l^{\prime
%}},b_{n}^{l^{\prime }}),\;l^{\prime }=l,l-1,...,1$ given by (\ref{D}), which
%is, so to say, "slow varying" and does not contribute to the leading term
%(the limit (\ref{ids})) of the corresponding eigenvalue distribution.

\section{Main Result and its Proof.}

As was already mentioned in Introduction, the goal of the paper is to extend
the results presented in \cite{Pe-Co:18} and justified in \cite{Pa:20} for
Gaussian weights and biases to those satisfying (\ref{bga}) -- (\ref{wga})
but not necessarily Gaussian. We will prove that in this
fairly general case the resulting eigenvalue distribution of random matrices
(\ref{JJM}) coincides with that of matrices of the same form where, however,
the analogs of diagonal matrices (\ref{D}), (\ref{kan}) are random
but \emph{independent} of $X^{l}$ (see Theorem \ref{t:main}, formulas (\ref{nukal}) -- (\ref{qlql}) in particular). Thus, we will comment first on the corresponding
result of random matrix theory (see, e.g. \cite%
{Pa:20,Pa-Sh:11,Co-Ha:14} and references therein).
%and follows, by the way, from our
%proof. % both known and proved in this paper.

Consider for every positive integer $n$: (i) the $n\times n$ random matrix $%
X_{n}$ with i.i.d. entries satisfying (\ref{wga}); (ii) positive definite
matrices $\mathsf{K}_{n}$ and $\mathsf{R}_{n}$
that are either deterministic or even random but independent of $X_{n}$ and such that their
Normalized Counting Measures $\nu _{\mathsf{K}_{n}}$ and $\nu _{\mathsf{R}%
_{n}}$ (see (\ref{ncm})) converge weakly as $%
n\rightarrow \infty $ (with probability 1 if random in an appropriate  probability space $\Omega_{\mathsf{KR}}$) to non-random measures $\nu _{\mathsf{K}}$ and $\nu _{%
\mathsf{R}}$. Set
\begin{equation}
\mathsf{M}_{n}=n^{-1}\mathsf{K}_{n}^{1/2}X_n\mathsf{R}_{n}X_{n}^{T}\mathsf{K}%
_{n}^{1/2}.  \label{mnbf}
\end{equation}%
According to random matrix theory (see, e.g. \cite{Pa:20,Pa-Sh:11,Co-Ha:14} and
references therein),
%and Theorem \ref{t:ind} below),
in this case and under certain conditions on $X_{n}$ the Normalized
Counting Measure $\nu _{\mathsf{%
M}_{n}}$ of $\mathsf{M}_{n}$ converges weakly with probability 1 as $%
n\rightarrow \infty $ (in the probability space $\Omega_{\mathsf{KR}} \times \Omega_{X}$, cf. (\ref{xinf})) to a non-random measure $\nu _{\mathsf{M}}$ which is
uniquely determined by the limiting measures $\nu _{\mathsf{K}}$ and $\nu _{%
\mathsf{R}}$ via a certain analytical procedure.
%(see, e.g. formulas (\ref{stm}) and (\ref{fhk}) -- (\ref{kh}) below).
We can write down this fact symbolically as %
\begin{equation}
\nu _{\mathsf{M}}=\nu _{\mathsf{K}}\diamond \nu _{\mathsf{R}}.  \label{opdia}
\end{equation}
%Moreover, it is easy to show that the role of $\nu _{K}$ and $\nu _{R}$
%can play any pair of non-negative measures having total mass 1 and a support
%belonging to the positive semiaxis. Thus,
In fact, the procedure defines a binary operation in the set of
non-negative measures with the total mass 1 and a support belonging to the
positive semi-axis (see more in Corollary \ref{c:conv}).
%and Remark \ref{r:diamond}).
%in particular, the operation can be repeated several
%times, every time with different pair $\nu _{K}$ and $\nu _{R}$.
%Note that the operation is just a version of the so-called multiplicative
%convolution of free probability theory \cite{Mi-Sp:17,Pe-Hi:00}, having the
%above random matrices as a basic analytic model.

The main result of works \cite{Pe-Co:18,Pa:20,Ya:20}, dealing with Gaussian
weights and biases and extended in this paper for any i.i.d. weights and
biases satisfying (\ref{wga}) and (\ref{bga})), is that the limiting
Normalized Counting Measure (\ref{ids}) of random matrices (\ref{JJM}),
where the role of $\mathsf{K}_{n}$ of (\ref{mnbf}) plays the matrix
defined by (\ref{D}) and (\ref{kan}) and depending on matrices $\{X^{l}\}^{L}_{l=1}$ of (\ref{wga}), is, nevertheless, equal to the "product" with respect the operation (\ref{opdia}) of
$L$ measures $\nu
_{K^{l}},\;l=1,...,L$ that
%are indicated in Theorem \ref{t:main} and
are the limiting Normalized Counting Measures of random matrices of (\ref{D})
and (\ref{kan}).
%\ given in (\ref{nukal}) -- (\ref{qlql}).
%that do not depend on $X^{l}$'s of (\ref{wga}),
%see (\ref{kab1}) and the subsequent text.

Note that the operation (\ref{opdia}) is closely related to the so-called multiplicative
convolution of free probability theory \cite{Mi-Sp:17}, having the
above random matrices as a basic analytic model.

Thus we will begin with a proof of this assertion which, we believe, is of independent interest for random matrix theory.

We follow \cite{Pe-Co:18,Pa:20} and confine ourselves to the case (\ref{nequ}%
) where all the weight matrices and bias vectors are of the same size $n$,
see (\ref{nequ}). The general case is essentially the same (see, e.g. Remark %
\ref{r:rra} (ii)).

In addition, we assume for the sake of simplicity of
subsequent formulas that
%the random matrices $X^{l}$ of (\ref{wga}) and the
%random vectors $b^{l}$ of (\ref{bga}) are i.i.d. in $l$ and we set
\begin{equation}
w^{2}=1  \label{wm4}
\end{equation}
in (\ref{wga}), thus, fixing the scale of the spectral axis. The general case
follows from the above by a simple change of variables.
%It is argued in \cite{Pe-Co:18} that the most
%interesting case from the point of view of
%machine learning is that where $%
%n_{l}=n,l=1,2,...,L$, hence the weights $W^{l},\;l=1,2,...,L$ are the i.i.d.
%$n\times n$ random matrices with i.i.d. independent entries of zero mean and
%unit variance and the biases $b^{l},$ $l=1,2,...,L$ are i.i.d. $n$%
%-dimensional random vectors with i.i.d. components of zero %mean and variance
%$\sigma _{b}^{2}>0$.
%We will begin with an assertion which show that the formula

\begin{theorem}
\label{t:ind} Consider for every positive integer $n$ the $n\times n$ random
matrix
\begin{equation}
\mathcal{M}_{n}=n^{-1}S_{n}X_{n}^{T}K_{n}X_{n}S_{n},  \label{mncal}
\end{equation}%
where:

\smallskip (a) $S_{n}$ is a positive definite $n\times n$ matrix such that
\begin{equation}
\sup_{n}n^{-1}\mathrm{Tr}R_{n}^{2}=r_{2}<\infty ,\;R_{n}=S_{n}^{2},
\label{r2}
\end{equation}%
and
\begin{equation}
\lim_{n\rightarrow \infty }\nu _{R_{n}}=\nu _{R},\;\nu _{R}(\mathbb{R}%
_{+})=1,  \label{nur}
\end{equation}%
where $\nu _{R_{n}}$ is the Normalized Counting Measure of $R_{n}$, $\nu
_{R} $ is a non-negative measure not concentrated at zero and $%
\lim_{n\rightarrow \infty }$ denotes the weak convergence of probability
measures (see \cite{Sh:96}, Section III.1);

\smallskip (b) $X_{n}$ is the $n\times n$ random matrix
\begin{equation}
X_{n}=\{X_{j\alpha }\}_{j,\alpha =1}^{n},\;\mathbf{E}\{X_{j\alpha }\}=0,\;%
\mathbf{E}\{X_{j\alpha }^{2}\}=1,\;\mathbf{E}\{X_{j\alpha
}^{4}\}=m_{4}<\infty  \label{Xn}
\end{equation}%
with jointly  i.i.d. entries (cf. (\ref{wga})), $b_{n}$ is the $n$%
-component random vector
\begin{equation}
b^{(n)}=\{b_{j}\}_{j=1}^{n},\;\mathbf{E}\{b_{j}\}=0,\;\mathbf{E}%
\{b_{j}^{2}\}=\sigma _{b}^{2}  \label{b}
\end{equation}%
with jointly  i.i.d. components (cf. (\ref{bga})) and for all $n$ the matrix $X_{n}$
and the vector $b_{n}$ viewed as defined on the same probability space%
\begin{equation}
\Omega _{X,b}=\Omega _{X}\times \Omega _{b},  \label{oxb}
\end{equation}%
where $\Omega _{X}$ and $\Omega _{b}$ are generated by the analogs of (\ref{xinf}) and (\ref%
{binf});

\smallskip (c) $K_{n}$ is the diagonal random matrix
\begin{equation}
K_{n}=\{\delta _{jk}K_{jn}\}_{j,k=1}^{n},\;K_{jn}=\left( \varphi ^{\prime }%
\Big(n^{-1/2}\sum_{a=1}^{n}X_{j\alpha }x_{\alpha n}+b_{j}\Big)\right) ^{2},
\label{Dn}
\end{equation}%
where $\varphi :\mathbb{R}\rightarrow \mathbb{R}$ is
continuously differentiable, is not identically constant and
such that (cf. (\ref{phi1}))
\begin{equation}
\sup_{x\in \mathbb{R}}|\varphi (x)|=\Phi _{0}<\infty ,\;\sup_{x\in \mathbb{R}%
}|\varphi ^{\prime }(x)|=\Phi _{1}<\infty ,  \label{Phi}
\end{equation}
$x_{n}=\{x_{\alpha n}\}_{\alpha =1}^{n}$ is a collection of real numbers
such that there exists the limit%
\begin{equation}
q=\lim_{n\rightarrow \infty }q_{n}>\sigma
_{b}^{2}>0,\;q_{n}=n^{-1}\sum_{\alpha =1}^{n}(x_{\alpha n})^{2}+\sigma
_{b}^{2}  \label{qqn}
\end{equation}
and
\begin{equation}
\lim_{n\rightarrow \infty }n^{-2}\sum_{\alpha =1}^{n}(x_{\alpha n})^{4}=0.
\label{xbou}
\end{equation}

Then the Normalized Counting Measure (NCM) $\nu _{\mathcal{M}_{n}}$ of $%
\mathcal{M}_{n}$ converges weakly with probability 1 in $\Omega _{X,b}$ of (%
\ref{oxb}) to a non-random measure $\nu _{\mathcal{M}}$, such that
%$\nu _{\mathcal{M}}(\mathbb{R_+})=1$%
\begin{equation}
\nu _{\mathcal{M}}(\mathbb{R_+})=1,  \label{numr1}
\end{equation}%
and that its Stieltjes transform $f_{\mathcal{M}}$ can be obtained from the
formulas%
\begin{align}
f_{\mathcal{M}}(z):& =\int_{0}^{\infty }\frac{\nu _{\mathcal{M}}(d\lambda )}{%
\lambda -z}  \notag \\
& =\int_{0}^{\infty }\frac{\nu _{R}(d\lambda )}{k(z)\lambda -z}%
=-z^{-1}+z^{-1}h(z)k(z), \; z \in \mathbb{C} \setminus \mathbb{R}_+, \label{fhk}
\end{align}%
where the pair ($h,k$) is the unique solution of the system of functional
equations%
\begin{equation}
h(z)=\int_{0}^{\infty }\frac{\lambda \nu _{R}(d\lambda )}{k(z)\lambda -z}, \; z \in \mathbb{C} \setminus \mathbb{R}_+
\label{hk}
\end{equation}%
\begin{equation}
k(z)=\int_{0}^{\infty }\frac{\lambda \nu _{K}(d\lambda )}{h(z)\lambda +1}, \; z \in \mathbb{C} \setminus \mathbb{R}_+,
\label{kh}
\end{equation}%
in which $\nu _{R}$ is defined in (\ref{nur}), and %$\nu _{K}$ is the
%probability distribution of
%$(\varphi ^{\prime }((q-\sigma _{b}^{2})^{1/2}\gamma
%+b_{1}))^{2}$ with $q$ of (\ref{qqn}) and the standard Gaussian random
%variable $\gamma $, i.e.,
\begin{equation}
\nu _{K}(\Delta )=\mathbf{P}\Big\{\big(\varphi ^{\prime }((q-\sigma
_{b}^{2})^{1/2}\gamma +b_{1})\big)^{2}\in \Delta \Big\},\;\Delta \in \mathbb{R},
\label{nuka}
\end{equation}%
where $q$ is given by (\ref{qqn}), $\gamma $ is the standard Gaussian random
variable and we are looking for a solution of (\ref{hk}) -- (\ref{kh})
%variable $\gamma $and we are looking for a solution of (\ref{hk}) -- (\ref{kh}) in the class
in the class of pairs $(h,k)$ of functions analytic outside
the closed positive semi-axis, continuous and positive on the negative semi-axis
and
such that\begin{equation}
\Im h(z)\Im z>0,\;\Im z\neq 0;\;\sup_{\xi \geq 1}\xi h(-\xi )\in (0,\infty ).
\label{hcond}
\end{equation}
\end{theorem}
The proof of the theorem is given in the next section. Here are the remarks.

\begin{remark}
\label{r:rra} (i) To apply Theorem \ref{t:ind} to the proof of Theorem \ref%
{t:main} we need a version of the former in which its "parameters",
i.e., $R_{n}$, hence $S_{n}$, in (\ref{mncal}) -- (\ref{nur}) and (possibly)
$\{x_{\alpha n}\}_{\alpha =1}^{n}$ in (\ref{Dn}) are
random, defined for all $n$ on the same probability space
$\Omega _{R,x}$,
independent of $\Omega _{X,b}$ of (\ref{oxb}) and satisfy conditions (\ref{r2}%
) -- (\ref{nur}) and (\ref{qqn}) -- (\ref{xbou}) with probability 1 in $%
\Omega _{R,x}$, i.e., on a certain subspace (cf. (\ref{px0}))
\begin{equation}  \label{px}
\overline{\Omega }_{R,x}\subset \Omega _{R,x},\;\mathbf{P}(\overline{\Omega
_{R,x}})=1.
\end{equation}
In this case Theorem \ref{t:ind} is valid with probability 1 in $\Omega
_{X,b}\times \Omega _{R,x}$. The corresponding argument is standard in random
matrix theory, see, e.g. Section 2.3 of \cite{Pa-Sh:11} and Remark \ref%
{r:penn} (iii).
%for its version in the present context.
%Let $\overline{\Omega }_{Xb}(\omega _{Rx})\subset $ $%
%\Omega _{Xb},\;\mathbf{P}(\overline{\Omega }_{Xb}(\omega _{Rx}))=1$ be the
%subspace of $\Omega _{Xb}$ of (\ref{oxb}) on which the theorem holds for %a given realization $\omega _{Rx}\in \overline{\Omega }_{Rx}$ of the
%parameters. Then it follows from the Fubini theorem that Theorem \ref{t:ind}
%holds on a certain $\overline{\Omega }\subset \Omega _{Rx}\times \Omega
%_{Xb},\;\mathbf{P}(\overline{\Omega })=1$.
The obtained limiting NCM\ $\nu _{\mathcal{M}} $ is random in general due to
the (possible) randomness of $\nu _{R}$ and $q$ in (\ref{nur}) and (\ref{qqn}%
) which are defined on their "own" probability space
$\Omega _{R,x}$ distinct from $\Omega_{X,b}$. %We will
%use this remark in the proof of Theorem \ref{t:main}.
Note, however, that in the case of Theorem \ref{t:main} the corresponding
analogs of $\nu _{R}$ and $q$ are not random, thus the limiting measure $\nu
_{M^{L}}$ is  non-random as well.

(ii) Repeating almost literally the proof of the theorem, one can treat a
more general case where $S_{m}$ is $m\times m$ positive definite matrix
satisfying (\ref{r2}) -- (\ref{nur}), $K_{n}$ is the $n\times n$ diagonal
matrix given by (\ref{Dn}) -- (\ref{qqn}), $X_{n}$ is a $n\times m$ random
matrix satisfying (\ref{wga}) and (cf. (\ref{asf1}))
\begin{equation*}
\lim_{m\rightarrow \infty ,n\rightarrow \infty }m/n=c\in (0,\infty ).
\end{equation*}%
%
% The corresponding modifications of the theorem are given
%in Remark \ref{r:rra1}(ii).
In this case the Stieltjes transform $f_{\mathcal{M}}$ of the limiting NCM
is again uniquely determined by three formulas where the first
and the second are (\ref{fhk}) and (\ref{hk}) with $k(z)$  replaced by $%
k(z)c^{-1}$ and the third coincides with (\ref{kh}).

(iii) Theorem \ref{t:ind} is proved above for bounded $\varphi$ and $\varphi'$
(see (\ref{Dn}) and (\ref{Phi})) and for the entries of $X_n$  and  the components of $b$ having the finite fourth and the second moment
respectively (see (\ref{Xn}) -- (\ref{b}). However, assuming the finiteness of these moments of sufficiently large order, it is possible to extend the theorem to the case where $\varphi$ and $\varphi'$ are just polynomially bounded. It suffices to apply to the matrices $K_n^l$ of (\ref{Dn}) a truncation procedure similar to that used for the matrix $R_n$ of (\ref{r2}), see formula (\ref{rrho}) and the subsequent text.
Correspondingly, Theorem \ref{t:main} can also be extended similarly, however in this case the maximal order of  finite moments depends on $L$.

(iv) The theorem provides the justification of the statistical independence of the random argument $n^{-1/2} (X_n x_n)_j + b_j$ of  (\ref{Dn}) and the weight matrix $n^{-1/2}X_n$ in (\ref{mncal}) in the infinite width limit $n \to \infty$. The assumption has been  used in a number of works (see e.g.
\cite{Pe-Co:18,Ma-Co:16,Po-Co:16,Pe-Ba:17,Sc-Co:17}) and is known as the mean field approximation, since it has certain similarity to the mean field approximation in
statistical mechanics and related fields.
\end{remark}

Theorem \ref{t:ind} yields an explicit form of the binary operation
(\ref{opdia}) via equations (\ref{fhk}) -- (\ref{kh}). Following \cite{Pa:20}%
, it is convenient to write the equations in a compact form similar to that
of free probability theory \cite{Mi-Sp:17}. This, in particular,
makes explicit the symmetry and the transitivity of the operation.

\begin{corollary}
\label{c:conv} Let $\nu _{K},\;\nu _{R}$ and $\nu _{\mathcal{M}}$ be the
probability measures (non-negative measures of the total mass 1)
entering (\ref{fhk}) -- (\ref{kh}) and $m_{K},\;m_{R}$ and $m_{\mathcal{M}}$
be their moment generating functions (see (\ref{mgen}) -- (\ref{stmg})).
Then their functional inverses $z_{K},\;z_{R}$ and $z_{\mathcal{M}}$ are
related as follows
\begin{equation}
z_{\mathcal{M}}(m)=z_{K}(m)z_{R}(m)m^{-1},  \label{mconv}
\end{equation}%
or, writing
\begin{equation}
z_{A}(m)=m\sigma _{A}(m),\;A=K,R,\mathcal{M,} \label{siz}
\end{equation}%
we obtain the simple "algebraic" form
\begin{equation}
\sigma _{\mathcal{M}}(m)=\sigma _{K}(m)\sigma _{R}(m)  \label{sconv}
\end{equation}
of the operation $\diamond$ of (\ref{opdia}).
\end{corollary}
%%%%%%%%%%%%%%%%%%%%%%%%%%%%%%%%%%%%%%%%%%%%%%%%%%%%%%%%%%%%%%%%%%%
\begin{proof}
It follows from (\ref{hk}) -- (\ref{kh}) and (\ref{stmg}) that%
\begin{align}
& m_{K}(-h(z))=-h(z)k(z),\;m_{R}(k(z)z^{-1})=-h(z)k(z),  \notag \\
& \hspace{1cm}m_{\mathcal{M}}(z^{-1})=-h(z)k(z).  \label{rels}
\end{align}%
Now the first and the third relations (\ref{rels}) yield $%
m_{K}(-h(z^{-1}))=m_{\mathcal{M}}(z)$, hence $z_{K}(m)=-h(z_{\mathcal{M}%
}^{-1}(m))$, and then the second and the third relations yield $%
m_{R}(k(z^{-1})z)=m_{\mathcal{M}}(z)$, hence $z_{R}(\mu )=k(z_{\mathcal{M}%
}^{-1}(m))z_{\mathcal{M}}(m)$. Multiplying these two relations and using
once more the third relation in (\ref{rels}), we obtain%
\begin{equation*}
z_{K}(m)z_{R}(m)=-k(z_{\mathcal{M}}^{-1}(m))h(z_{\mathcal{M}}^{-1}(m))z_{%
\mathcal{M}}(m)=z_{\mathcal{M}}(m)m
\end{equation*}%
and (\ref{mconv}) -- (\ref{sconv}) follows.
\end{proof}
%%%%%%%%%%%%%%%%%%%%%%%%%%%%%%%%%%%%%%%%%%%%%%%%%%%%%%%%%%%%%%%%%%%%%%%
\begin{remark}
\label{r:diamond}  In the case of rectangular matrices $X_{n}$ in (\ref{mncal}%
), described in Remark \ref{r:rra} (ii), the analogs of (\ref{mconv}) and (%
\ref{sconv}) are
\begin{equation}
z_{\mathcal{M}}(m)=z_{K}(cm)z_{R}(cm)m^{-1},\;\sigma _{\mathcal{M}%
}(m)=c^{2}\sigma _{K}(cm)\sigma _{R}(cm).  \label{rect}
\end{equation}
%(ii) It is also worth noting, following again free probability theory \cite%
%{Ch-Co:18,Mi-Sp:17}, that formulas (\ref{hk}) -- (\ref{kh}) as well as (\ref%
%{mconv} -- (\ref{sconv}) can be viewed as those determining a binary operation %acting on
%measures with support in $[0,\infty )$ and producing a probability measure $%
%\nu _{12}$ out of measures $\nu _{1}$ and $\nu _{2}$ (not necessarily those
%%describing limiting eigenvalue distribution of certain large size matrices).
%Indeed, given measures $\nu _{1}$ and $\nu _{2}$, we find via \ (\ref{fh})
%-- (\ref{kh}) with $\nu _{K}=\nu _{1},\;\nu _{2}=\nu _{R}$ a unique measure $%
%\nu _{\mathcal{M}}=\nu _{12}$ whose Stieltjes transform is given by (\ref{fh}). %The operation is denoted
%\begin{equation}
%\nu _{12}=\nu _{1}\diamond \nu _{2}  \label{conv}
%\end{equation}%
%and is a version of the free multiplicative convolution %\cite{Ch-Co:18,Mi-Sp:17}.
%(ii) It follows from (\ref{sconv}) that in is the algebraic form
\end{remark}

\medskip
We will now formulate and prove our main result.

\begin{theorem}
\label{t:main} Let $M_{n}^{L}$ be the random matrix (\ref{JJM}) defined by (%
\ref{rec}) -- (\ref{D}) and (\ref{nequ}), where the weights $%
\{W^{l}\}_{l=1}=\{n^{-1/2}X^{l}\}_{l=1}$ and biases $\{b^{l}\}_{l=1}$ are
i.i.d. in $l$ with i.i.d. entries and components satisfying (\ref{bga}) -- (%
\ref{wga}) and the input vector $x^{0}$ (\ref{x0}) (deterministic or random)
is such that there exists a finite limit%
\begin{equation}
q^{1}:=\lim_{n\rightarrow \infty }q_{n}^{1}>\sigma
_{b}^{2}>0,\;\;q_{n}^{1}=n^{-1}\sum_{j_{0}=1}^{n}(x_{j_{0}}^{0})^{2}+\sigma
_{b}^{2}  \label{q0}
\end{equation}%
and
\begin{equation}
\lim_{n\rightarrow \infty }n^{-2}\sum_{j_{0}=1}^{n}(x_{j_{0}}^{0})^{4}=0.
\label{qLin}
\end{equation}%
Assume also that the activation function $\varphi $ in (\ref{rec}) is
continuously differentiable, $\varphi ^{\prime }$ is not
zero identically and%
\begin{equation}
\sup_{t\in \mathbb{R}}|\varphi (t)|=:\Phi _{0}<\infty ,\;\sup_{t\in \mathbb{R}%
}|\varphi ^{\prime }(t)|=:\Phi _{1}<\infty .\;\;  \label{phi1}
\end{equation}%
Then the Normalized Counting Measure (NCM) $\nu _{M_{n}^{L}}$ of $M_{n}^{L}$
(see (\ref{ncm})) converges weakly with probability 1 in the probability
space $\Omega _{L}$ of (\ref{oml}) to a non-random limit
\begin{equation}
\nu _{M^{L}}=\lim_{n \to \infty}   \nu _{M^{L}_n},  \label{numl}
\end{equation}
where
\begin{equation}
\nu _{M^{L}}=\nu _{K^{L}}\diamond \cdots \diamond \nu
_{K^{1}},  \label{nucon}
\end{equation}%
the operation "$\diamond $" is defined in (\ref{opdia}) (see also
%Lemma 3.5 in \cite{Pa:20} and
Remark \ref{r:penn} (ii) Corollary \ref{c:conv} below) and
\begin{equation}
\nu _{K^{l}}(\Delta )=\mathbf{P}\Big\{\big(\varphi ^{\prime }((q^{l}-\sigma
_{b}^{2})^{1/2}\gamma +b_{1})\big)^{2}\in \Delta \Big\},\;\Delta \in \mathbb{R}%
,\;l=1,...,L,  \label{nukal}
\end{equation}%
with the standard Gaussian random variable $\gamma $ and $q^{l}$ determined
by the recurrence
\begin{equation}
q^{l}=\int \varphi ^{2}\Big(\gamma (q^{l-1}-\sigma _{b}^{2})^{1/2}+b\Big)%
\Gamma (d\gamma )F(db),\;l\geq 2,  \label{qlql}
\end{equation}%
where $\Gamma (d\gamma )=(2\pi )^{1/2}e^{-\gamma ^{2}/2}d\gamma $ is the
standard Gaussian measure, $F$ is the probability law of $b_{1}^{l}$ in (\ref%
{bga}) and $q^{1}$ is given by (\ref{q0}).
\end{theorem}

\begin{remark}
\label{r:penn} (i) If
\begin{equation}
q_{1}=\cdots =q_{L}=:q^*,  \label{qeq}
\end{equation}%
then $\nu _{K}:=\nu _{K^{l}},\;l=1,\dots,L$ and (\ref{nucon}) becomes
\begin{equation}
\nu _{M^{L}}=\underset{L\;\mathrm{times}}{\underbrace{\nu _{K}\diamond \nu
_{K}\cdots \diamond \nu _{K}}}. \label{nucone}
\end{equation}%
%and the moment generating function (\ref{mgen}) of $\nu _{M^{L}}$ satisfies
%functional equation (\ref{penfo}) proposed in \cite{Pe-Co:18}.
Equalities (\ref{qeq}) are the case if $q^{\ast }$ is a fixed point of
(\ref{qlql}), see \cite{Ma-Co:16,Po-Co:16,Sc-Co:17} for a detailed analysis
of (\ref{qlql}) with Gaussian $F$ and its role in the functioning of the
deep neural networks.

(ii) Let us show  that Theorem \ref{t:main} implies the
results of   \cite{Pe-Co:18}, formula (\ref{penfo}) in particular. Indeed, it follows from the
theorem, (\ref{mlml1}), and Corollary \ref{c:conv} that the functional inverse $z_{M^{l+1}}$ of the
moment generating function $m_{M^{l+1}}$ (see (\ref{mgen}) -- (\ref{stmg}))
of the limiting NCM $\nu _{M^{l+1}}$ of matrix $M_{n}^{l+1}$ and that of $%
M_{n}^{l}$ are related as (cf. (\ref{mconv}))
\begin{equation}
z_{M^{l+1}}(m)=z_{K^{l+1}}(m)z_{M^{l}}(m)m^{-1}.  \label{zre}
\end{equation}%
Passing from the moment generating functions to the S-transforms of free
probability theory via the formula $S(m)=z(m)(1+m)m^{-1}$  \cite{Mi-Sp:17} and taking into
account that for the limiting NCM (\ref{nump}) of the Wishart matrix $%
n^{-1}X_{n}X_{n}^{T}$, we have $z_{MP}(m)=m(1+m)^{-2}$ and
$S_{MP}(m)=(1+m)^{-1}$. This and  (\ref{zre}) imply%
\begin{equation}
S_{M^{l+1}}(m)=S_{K^{l+1}}(m)S_{MP}(m)S_{M^{l}}(m),  \label{prec}
\end{equation}%
 another form of the operation (\ref{opdia}), \ cf. (\ref{siz}).
Next, iterating (\ref{prec})  $L$ times, we obtain
% $S_{M^{L}}=S_{K^{L}}\cdots (m)S_{MP}(m)S_{M^{l}}(m)$
% obtain in the case
%(\ref{qeq}) $L$ times,
 again (\ref{nucone}), and iterating (\ref{zre})  $L$ times under condition (\ref{qeq}), we obtain
\[z_{K}(m)=(z_{M^L}^L (m))^{1/L}(1+m)^{1/L}m^{(L-1)/L}
\]
implying  (\ref{penfo}) (formula (13) of
\cite{Pe-Co:18}).
%The functional equation (\ref{penfo}) arising in the case (%
%\ref{qeq}) of the $l$-independent parameters $q_{l}$ of (\ref{qlql}) is%
%derived from (\ref{prec}) in \cite{Pe-Co:18}.
%Formula (\ref{nucon}) is a version of formula (13) in \cite{Pe-Co:18}
%(see also \cite{Mu:02}), although our operation $\diamond$ of (\ref{opdia})
%is somewhat different from the free multiplicative convolution used in %\cite{Pe-Co:18}, see Remark \ref{r:pefr}.

(iii) If the input vectors (\ref{x0}) are random, then it is necessary to
assume that they are defined on the same probability space $\Omega _{x^{0}}$
for all $n_{0}$ and that (\ref{q0}) -- (\ref{qLin}) are valid with probability 1
in $\Omega _{x^{0}}$, i.e., there exists
\begin{equation}  \label{px0}
\overline{\Omega}_{x^{0}} \subset \Omega _{x^{0}}, \; \mathbf{P}(\overline{%
\Omega }_{x^{0}})=1
\end{equation}
where (\ref{q0}) and (\ref{qLin}) hold. It follows then from the Fubini
theorem that in this case the set $\overline{\Omega}_L \subset \Omega_L,
\mathbf{P}\{\overline{\Omega}_L\}=1$ where Theorem \ref{t:main} holds has to
be replaced by the set $\overline{\Omega}_{Lx^{0}} \subset \Omega_L \times
\Omega _{x^{0}}, \; \mathbf{P} \{\overline{\Omega}_L\}=1$. An example of
this situation is where $\{x_{j_{0}}^{0}\}_{j^{0}=1}^{n_{0}}$ are the first $%
n_{0}$ components of an ergodic sequence $\{x_{j_{0}}^{0}\}_{j^{0}=1}^{%
\infty }$ (e.g. a sequence of i.i.d. random variables) with finite fourth
moment. Here $q^{1}$ in (\ref{q0}) exists with probability 1 on the
corresponding $\Omega _{x^{0}}$ and even is non-random just by ergodic
theorem (the strong Law of Large Numbers in the case of i.i.d. sequence),
the r.h.s. of (\ref{qLin}) is $n^{-1}\mathbf{E}\{(x_{1}^{0})^{4}\}(1+o(1)),%
\;n\rightarrow \infty $ with probability 1 in $\Omega _{x^0}$ and the
theorem is valid with probability 1 in $\Omega _{L}\times \Omega _{x^{0}}$.

(iv) An analog of Theorem \ref{t:main} corresponding to the more general
case (\ref{asf1}) is also valid. It suffices to use Remark \ref{r:rra} (ii).
Likewise, conditions (\ref{phi1}) can also be replaced by those requiring
polynomial bounds for $\varphi$ and $\varphi'$ provided that the components of the
bias vectors (\ref{bga}) and the entries of the weight matrices (\ref{wga})
have finite moments of sufficiently high order which may depend on $L$,
see Remark \ref{r:rra} (iii).
\end{remark}
We present now the proof of Theorem \ref{t:main}. \medskip
\begin{proof}
We prove the theorem by induction in $L$. We have from (\ref{rec}) -- (\ref%
{JJM}) and (\ref{nequ}) with $L=1$ the following $n\times n$ random matrix%
\begin{equation}
M_{n}^{1}=J_{n}^{1}(J_{n}^{1})^{T}=n^{-1}D_{n}^{1}X_{n}^{1}(X_{n}^{1})^{T}D_{n}^{1}.
\label{m1}
\end{equation}%
It is convenient to pass from $M_{n}^{1}$ to the $n\times n$ matrix (see
Remark \ref{r:kgen})%
\begin{equation}
\mathcal{M}%
_{n}^{1}=(J_{n}^{1})^{T}J_{n}^{1}=n^{-1}(X_{n}^{1})^{T}K_{n}^{1}X_{n}^{1},%
\;K_{n}^{1}=(D_{n}^{1})^{2}  \label{cm1}
\end{equation}%
which has the same spectrum, hence the same Normalized Counting Measure as $%
M_{n}^{1}$. The matrix $\mathcal{M}_{n}^{1}$ is a particular case with $%
S_{n}=\mathbf{1}_{n}$ of matrix (\ref{mncal}) treated in Theorem \ref{t:ind}
above. Since the NCM of $\mathbf{1}_{n}$ is the Dirac measure $\delta _{1}$,
conditions (\ref{r2}) -- (\ref{nur}) of the theorem are evident. Conditions (%
\ref{qqn}) and (\ref{xbou}) of the theorem are just (\ref{q0}) and (\ref%
{qLin}). It follows then from Corollary \ref{c:conv} that the assertion of
our theorem, i.e., formula (\ref{nucon}) with $q^{1}$ of (\ref{q0}) is valid
for $L=1$.

Consider now the case  where $L=2$ of (\ref{rec}) -- (\ref{JJM}) and (\ref{nequ}):
\begin{equation}
M_{n}^{2}=n^{-1}D_{n}^{2}X_{n}^{2}M_{n}^{1}(X_{n}^{2})^{T}D_{n}^{2}.
\label{m2m1}
\end{equation}%
Since $M_{n}^{1}$ is positive definite, we write
\begin{equation}
M_{n}^{1}=(S_{n}^{1})^{2}  \label{m1s1}
\end{equation}%
with a positive definite $S_{n}^{1}$, hence,
\begin{equation}
M_{n}^{2}=n^{-1}D_{n}^{2}X_{n}^{2}(S_{n}^{1})^{2}(X_{n}^{2})^{T}D_{n}^{2}
\label{mn2}
\end{equation}%
and the corresponding $\mathcal{M}_{n}^{2}$ is
(see Remark \ref{r:kgen})
\begin{equation}
\mathcal{M}%
_{n}^{2}=n^{-1}S_{n}^{1}(X_{n}^{^{2}})^{T}K_{n}^{2}X_{n}^{2}S_{n}^{1},%
\;K_{n}^{2}=(D_{n}^{2})^{2}.  \label{cm1m2}
\end{equation}%
We observe that $\mathcal{M}_{n}^{2}$ is a particular case of matrix (\ref%
{mncal}) of Theorem \ref{t:ind} with $M_{n}^{1}$ of (\ref{m1s1}) as $%
R_{n}=(S_{n})^{2}$, $X_{n}^{2}$ as $X_{n}$, $K_{n}^{2}$ as $K_{n}$, $%
\{x_{j_{1}}^{1}\}_{j_{1}=1}^{n}$ as $\{x_{\alpha n}\}_{\alpha =1}^{n}$, $%
\Omega _{1}=\Omega ^{1}$ of (\ref{oml}) as $\Omega _{R,x}$ and $\Omega ^{2}$
of (\ref{oml}) as $\Omega _{X,b}$, i.e., the case of the random but\ $%
\{X_{n}^{2},b^{2}_{n}\}$ -independent $R_{n}$ and $\{x_{\alpha n}\}_{\alpha
=1}^{n}$ in (\ref{mncal}) as described in Remark \ref{r:rra} (i). Let us
check that conditions (\ref{r2}) -- (\ref{nur}) \ and (\ref{qqn}) -- (\ref%
{xbou}) of Theorem \ref{t:ind} are satisfied for $\mathcal{M}_{n}^{2}$ of (%
\ref{cm1m2}) with probability 1 in the probability space $\Omega
_{1}=\Omega^1$ generated by $\{X_{n}^{1},b_{n}^{1}\}$ for all $n$ and
independent of the space $\Omega ^{2}$ generated by $\{X_{n}^{2},b_{n}^{2}\}$
for all $n$.

To this end we use an important fact on the operator norm of $n\times n$
random matrices with i.i.d. entries satisfying (\ref{wga}). Namely, if $%
X_{n} $ is a such $n\times n$ matrix, then we have with probability 1
\begin{equation}
\lim_{n\rightarrow \infty }n^{-1/2}||X_{n}||=2,  \label{lin}
\end{equation}%
thus, with the same probability%
\begin{equation}
||X_{n}||\leq Cn^{1/2},\;2 \le C < \infty,\; n \ge n_0,  \label{nox}
\end{equation}%
if $n_0$ is large enough.

For the Gaussian matrices relation (\ref{lin}) has already been known in the
Wigner school of the early 1960th, see \cite{Pa-Sh:11}. It follows in this
case from the orthogonal polynomial representation of the density of the NCM
of $n^{-1}X_{n}X_{n}^{T}$ and the asymptotic formula for the corresponding
orthogonal polynomials. For the modern form of (\ref{lin}),
in particular its validity for random matrices with i.i.d entries of zero
mean and finite fourth moment, see \cite{Ba-Si:00,Ve:18} and references
therein.

We will also need the bound%
\begin{equation}
||K_{n}^{1}||\leq \Phi^{2} _{1},  \label{nok}
\end{equation}%
following from (\ref{D}), (\ref{kan}) and (\ref{phi1}) and valid everywhere
in $\Omega _{1}$ of (\ref{oml}).

Now, by using (\ref{cm1}), (\ref{nox}), (\ref{nok}) and the inequality
\begin{equation}
|\mathrm{Tr}AB|\leq ||A||\mathrm{Tr}B,\;  \label{tab}
\end{equation}%
valid for any matrix $A$ and a positive definite matrix $B$, we obtain with
probability 1 in $\Omega _{1}$ and for sufficiently large $n_{0}$ of (\ref{nok})
\begin{equation*}
n^{-1}\mathrm{Tr}(M_{n}^{1})^{2}=n^{-3}\mathrm{Tr}\mathbb{(}%
K_{n}^{1}X_{n}^{1}(X_{n}^{1})^{T})^{2}\leq (C\Phi _{1})^{4}.
\end{equation*}%
We conclude that $M_{n}^{1}$, which plays here the role of $R_{n}$ of
Theorem \ref{t:ind} and Remark \ref{r:rra} (i) according to (\ref{m1s1}),
satisfies condition (\ref{r2}) with $r_{2}=(C\Phi _{1})^{4}$ and with
probability 1 in our case, i.e., on a certain $\Omega _{11}\subset \Omega
_{1},\;\mathbf{P}(\Omega _{11})=1$.

Next, it follows from the above proof of the theorem for $L=1$, i.e., in
fact, from Theorem \ref{t:ind}, that there exists $\Omega _{12}\subset
\Omega _{1},\;\mathbf{P}(\Omega _{12})=1$ on which the NCM $\nu _{M_{n}^{1}}$
converges weakly to a non-random limit $\nu _{M^{1}}$, hence condition (\ref%
{nur}) is also satisfied with probability 1, i.e., on $\Omega _{12}$.

At last, according to Lemma \ref{l:xlyl} and (\ref{q0}), there exists $%
\Omega _{13}\subset \Omega _{1},\;\mathbf{P}(\Omega _{13})=1$ on which there
exists
\begin{equation*}
\lim_{n\rightarrow \infty
}n^{-1}\sum_{j_{0}=1}^{n}(x_{j_{0}}^{0})^{2}+\sigma _{b}^{2}=q^{1}>\sigma
_{b}^{2},
\end{equation*}%
and according to (\ref{rec}) and (\ref{phi1}) we have uniformly in $n$: $%
|x_{j_{1}}^{1}|\leq \Phi _{0},\;j_{1}=1,...,n$, i.e., conditions (\ref{qqn}) and (\ref{qLin})
are also satisfied.

Hence, we can apply Theorem \ref{t:ind}\ on the subspace $\overline{\Omega }%
_{1}=\Omega _{11}\cap \Omega _{12}\cap \Omega _{13}\subset \Omega _{1},\;%
\mathbf{P}(\overline{\Omega }_{1})=1$ where all the conditions of the
theorem are valid, i.e., $\overline{\Omega }_{1}$ plays the role of $\Omega
_{R,x} $ of Remark \ref{r:rra} (i). Then, the theorem implies that for every $%
\omega _{1}\in \overline{\Omega }_1$ there exists subspace $\overline{%
\Omega^2}(\omega _{1})$ of the space $\Omega ^{2}$  generated by $%
\{X_{n}^{2},b_{n}^{2}\}$ for all $n$ and such that $\mathbf{P}(\overline{%
\Omega^2}(\omega _{1}))=1$ and formulas (\ref{nucon}) -- (\ref{qlql}) are
valid for $L=2$. It follows then from the Fubini theorem that the same is
true on a certain $\overline{\Omega }_{2}\subset \Omega _{2},\;\mathbf{P}(%
\overline{\Omega }_{2})=1$ where $\Omega _{2}$ is defined by (\ref{oml})
with $L=2$.

This proves the theorem for $L=2$. The proof for $L=3,4,\dots$ is analogous,
since (cf. (\ref{mn2}))
\begin{equation}
M_{n}^{l+1}=n^{-1}D_{n}^{l+1}X_{n}^{l+1}M_{n}^{l}(X_{n}^{l+1})^{T}D_{n}^{l+1},\;l\geq 2.
\label{mlml1}
\end{equation}%
In particular, we have with probability 1 on $\Omega _{l}$ of (\ref{oml})
for $M_{n}^{l}$ playing the role of $R_{n}$ of Theorem \ref{t:ind} on the $l$%
th step of the inductive procedure (cf. (\ref{r2}))
\begin{equation*}
n^{-1}\mathrm{Tr}(M^{l})^{2}\leq (C\Phi _{1})^{4l},\;l\geq 2.
\end{equation*}
If $x_0$ is random, then it is necessary to follow the argument
given in Remark \ref{r:penn} (iii).
\end{proof}
%\smallskip
Note that the material of this section is quite close to that
of Section 2 of \cite{Pa:20}.
\medskip

We will comment now on our numerical results presented on Fig. 1 -- Fig. 4 below.
The figures, except Fig. 1a), show the arithmetic mean $\rho _{n}$ of the empirical
eigenvalue densities of a certain number $N$ of samples of $%
M_{n}^{L}$ with various layer widths $n$, network depths $L$ and
activation functions $\varphi $. The entries of the weight matrices
$W^{l}$ and the components of bias vectors $b^{l}$ of
(\ref{rec}) -- (\ref{bl}) are Gaussian satisfying
(\ref{bga}) \ -- (\ref{wga}) for Fig. 1 -- Fig. 3
and the Cauchy random variables with the density
\begin{equation}\label{cauch}
p(x)=\frac{\delta}{\pi (x^2 + \delta^2)}, \; \delta=1/n
\end{equation}
for Fig. 4.
The number $N$ is roughly the minimum
number of samples providing stable (reproducing) numerical
results for $\rho _{n}$ such that the theoretical curve obtained
numerically from (\ref{nucon}) -- (\ref{qlql}) and (\ref{fhk}) -- (\ref{nuka})
either coincides (within the accuracy of our numerical simulations) with $\rho _{n}$
or is its smooth version.
%for Fig. 1 -- Fig. 3.
We have used
\begin{equation}\label{Nn}
N=10^{7} \; \mathrm{for} \; n=10,\;30, \;\; N=10^{6} \; \mathrm{for} \; n=10^{2}, \;\; N=10^{4} \;
\mathrm{for} \; n=10^3.
\end{equation}

\begin{figure}[ht]
\center{\includegraphics[width=1.0\linewidth]{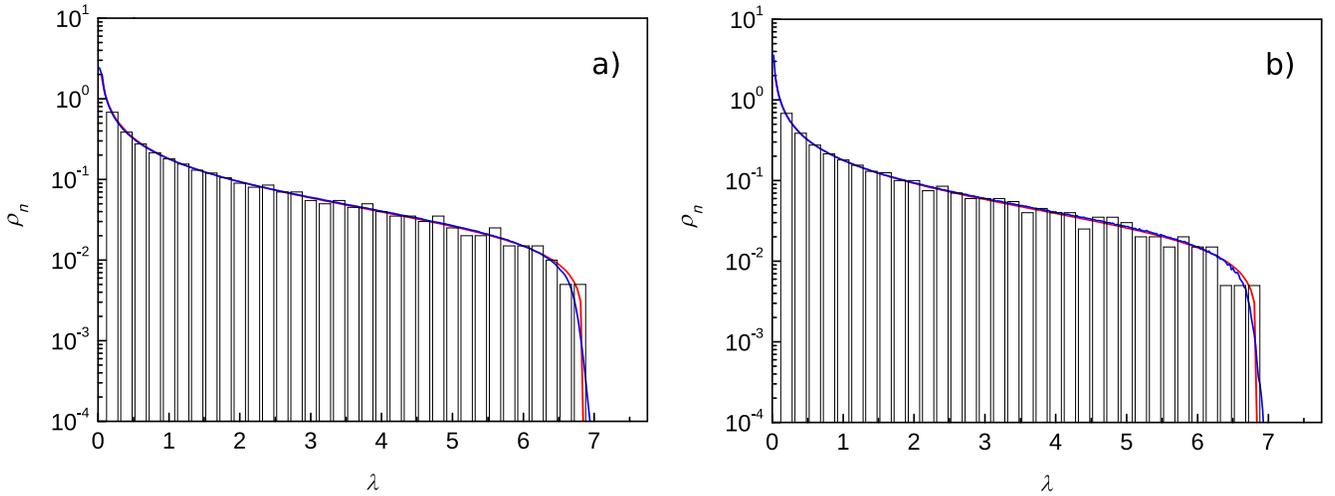} }
\caption{\small
%Singular values distribution of (\ref{JJM}). Network depth $L=2$, layer width %$n=1000$.
%Solid blue curves correspond to averaging over $N=10^3$ realization. Bar %``histogram'' plots correspond to distribution of (\ref{JJM}) singular values on %a realization.
%a) Linear activation function, b) ``Hard-Tanh'' activation function.
The eigenvalue density (in the semi-log scale) of the random matrix $M_{n}^{L}$ (\ref{JJM}%
) for the Gaussian weights and biases. The network depth $L=2$ and
the layer width $n=10^{3}$. The histograms correspond to
the density of a sample of $M_{10^{3}}^{2}$, the
solid blue curves to the arithmetic means $\rho _{n}$ of the
sample densities of $N=10^{3}$ samples of $M_{10^{3}}^{2}$ and the solid
red line to the numerical solution of equations (\ref{nucon}) -- (\ref{qlql}). a) The
linear activation function (conventional random matrix theory); b) the HardTanh activation function (\ref{hata}).
}
\label{fig1}
\end{figure}

Figure 1. The eigenvalue densities of the random matrix $M_{n}^{L}$ (\ref{JJM}%
) for $L=2$ and $n=10^{3}$.
%and for the Gaussian weight matrices and the bias vectors.
The histograms are obtained from the sample of
$M_{10^{3}}^{2}$, the solid blue curves is the plot of the arithmetic
means $\rho _{n}$ over $N=10^{3}$ samples of $M_{10^{3}}^{2}$ and the solid
red curves are the result of the numerical solutions of equation (\ref{nucon}) -- (\ref{qlql}), where the operation $diamond$ is given by (\ref{fhk}) -- (\ref{kh}).
a) Linear activation function, b) The HardTanh activation function, i.e.,
\begin{align}
%\notag
%& a) \; \; \varphi (x)=x,\\
%& b)  \; \; \varphi   \label{hata}
%=(x+1)_{+}-(x-1)_{+}-1
%=\left\{
%\begin{array}{cc}
%-1, & x\leq -1, \\
%x, & |x|\leq 1, \\
%+1, & x\geq 1.
%\end{array}%
%\right.
a) \; \varphi (x)=x, \; \; \;
b) \; \varphi   \label{hata}
%=(x+1)_{+}-(x-1)_{+}-1
=\left\{
\begin{array}{cc}
-1, & x\leq -1, \\
x, & |x|\leq 1, \\
+1, & x\geq 1.
\end{array}%
\right.
\end{align}%
%where $x_{+}=\max \{x,0\}$.
The figure demonstrates the quite good fitting of
the three descriptions of the eigenvalue density, thereby manifesting the
fast convergence of the numerically obtained results to
to a non-random limit given by Theorems \ref{t:ind} -- \ref{t:main}.
%(\ref{numr1}).

Figure 2.  a) displays the density $\nu _{K}^{\prime }$ of the measure $%
\nu _{K}$ of (\ref{nukal}) for the indicated activation functions $\varphi $.
%and Gaussian weights and biases, see (\ref{rec}) -- (\ref{bl}) and (\ref%
%{bga}) -- (\ref{wga}).
It is well seen that all (except $\varphi (x)=x$) activation functions
lead to quite similar $\nu
_{K}^{\prime }$ having two narrow peaks centered at $0$ and $1$
and being rather close to zero otherwise. This is natural (in fact,
exact) for the HardTanh (\ref{hata}), seems
likely for the smooth sigmoid $\varphi (x)=\tanh x$, less likely for $%
\varphi (x)=\sinh x$ and rather surprising for $\varphi
(x)=\sin x$. Nevertheless, according to Fig. 2b), the mean
eigenvalue densities $\rho _{n} $ obtained from $N=10^{4}$ samples
of $M_{10^{3}}^{2}$ for all $\varphi
$  including $\varphi
(x)=x$ are very close within the (semi-log) scale of the figure. %It is
%tempting to believe that the form of $\nu _{K}^{\prime }$, hence,
%$\rho _{n}$ (and $\nu _{M^{L}}^{\prime }$) is essentially
%determined by certain neighborhoods of minima and maxima of
%$(\varphi ^{\prime })^{2}.$

The weak dependence of $\rho _{n}$ on $\varphi $ can be viewed as an analog
of the macroscopic universality (the universality of the global
regime) in random matrix theory \cite{Pa-Sh:11,Ta-Vu:10}, where the
limiting eigenvalue distribution of the Wigner matrices and the
sample covariance matrices are completely determined just by the
second moment of the matrix entries. An analog of this type
universality is also proved in this paper (we believe that
condition on the fourth moment in (\ref{wga}) can be removed).

\begin{figure}[ht]
\center{\includegraphics[width=1.0\linewidth]{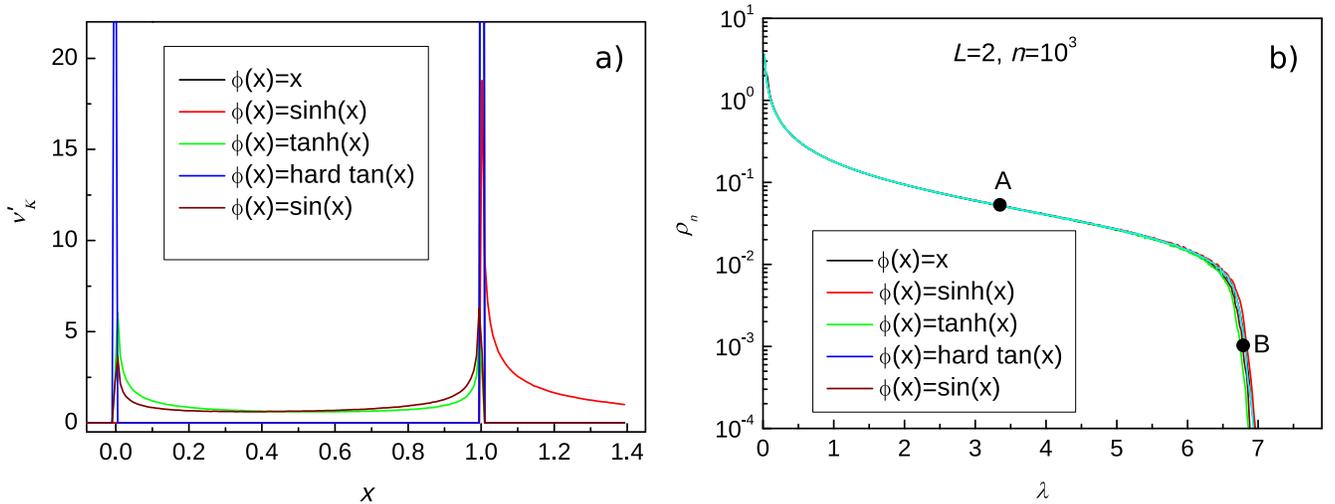} }
\caption{\small
%a) Distribution of $\phi'(x)^2$ for different activation functions. b) Singular %values distributions $P(s)$ of (\ref{JJM}) for different $\phi(x)$. Network %depth $L=2$, layer width $n=1000$.
%The maximal distance between curves in the vicinity
%of {\bf A} is about 1.3\%, %and in the vicinity of {\bf B}
%is about 0.7\%.
%
a) The density $\nu _{K}^{\prime }$ of the measure $\nu _{K}$ of (%
\ref{nukal}) for the indicated activation functions and the Gaussian
weights and biases. b) The arithmetic means $\rho _{n}$ (in the semi-log scale) of the
sample eigenvalue
densities of $M_{10^{3}}^{2}$ over $N=10^{4}$ samples for all indicated $%
\varphi $. %In a more refined scale the maximum distance between
%the curves corresponding to different $\varphi $'s is 2.1\% in the
%neighborhood of $A$ and 1.5\% in the neighborhood of $B$.
}
\label{fig2}
\end{figure}

Note, however, that the "universality" shown in Fig. 2
%(observed and discussed in \cite{Pe-Co:18})
is not exact. Indeed, by using a more refined scale for the
curves of Fig. 2 b), we found that curves differ by $2.1\%$ in a
neighborhood of $\ A$ and by $1.5\%$ in a neighborhood of $B$. In addition,
it follows from Theorem \ref{t:ind} that if $\nu _{K^{\prime }}\neq \nu
_{K^{\prime \prime }}$, then $\nu _{(M^{L})^{\prime }}$ $\neq \nu
_{(M^{L})^{\prime \prime }}$. However, this fact could be of interest for
applications, since it implies that up to a certain precision one can
confine oneself to a simple case of linear, i.e., the standard random
matrix, results and calculations. It is instructive in this context to
%Moreover, given a non-linear (sigmoid or
%not) bounded and continuous ,
consider the following family of activation functions:%
\begin{equation}
\varphi _{\varepsilon }(x)=\varepsilon ^{-1}\psi (\varepsilon x),
\label{phiep}
\end{equation}%
where $\psi:\mathbb{R}\rightarrow \mathbb{R},\;\psi
(x)=x+o(x),\;x\rightarrow 0$ (sigmoid or
not) is bounded and continuous. We have then for any $x$:
%\begin{equation*}
$\varphi _{\varepsilon }(x)=x(1+o(1)),\;\varphi _{\varepsilon }^{\prime
}(x)=1+o(1),\;\varepsilon \rightarrow 0$.
%\end{equation*}%
We conclude that the linear case $\varphi (x)=x$ can be viewed as an
asymptotic regime for the family (\ref{phiep}). It is
remarkable, however, that according to Fig. 2 the regime seems to be
applicable up to $\varepsilon \simeq 1$. An analogous property was found in \cite{Ba-Ke:20}
although in a different context.

\begin{figure}[ht]
\center{\includegraphics[width=1.0\linewidth]{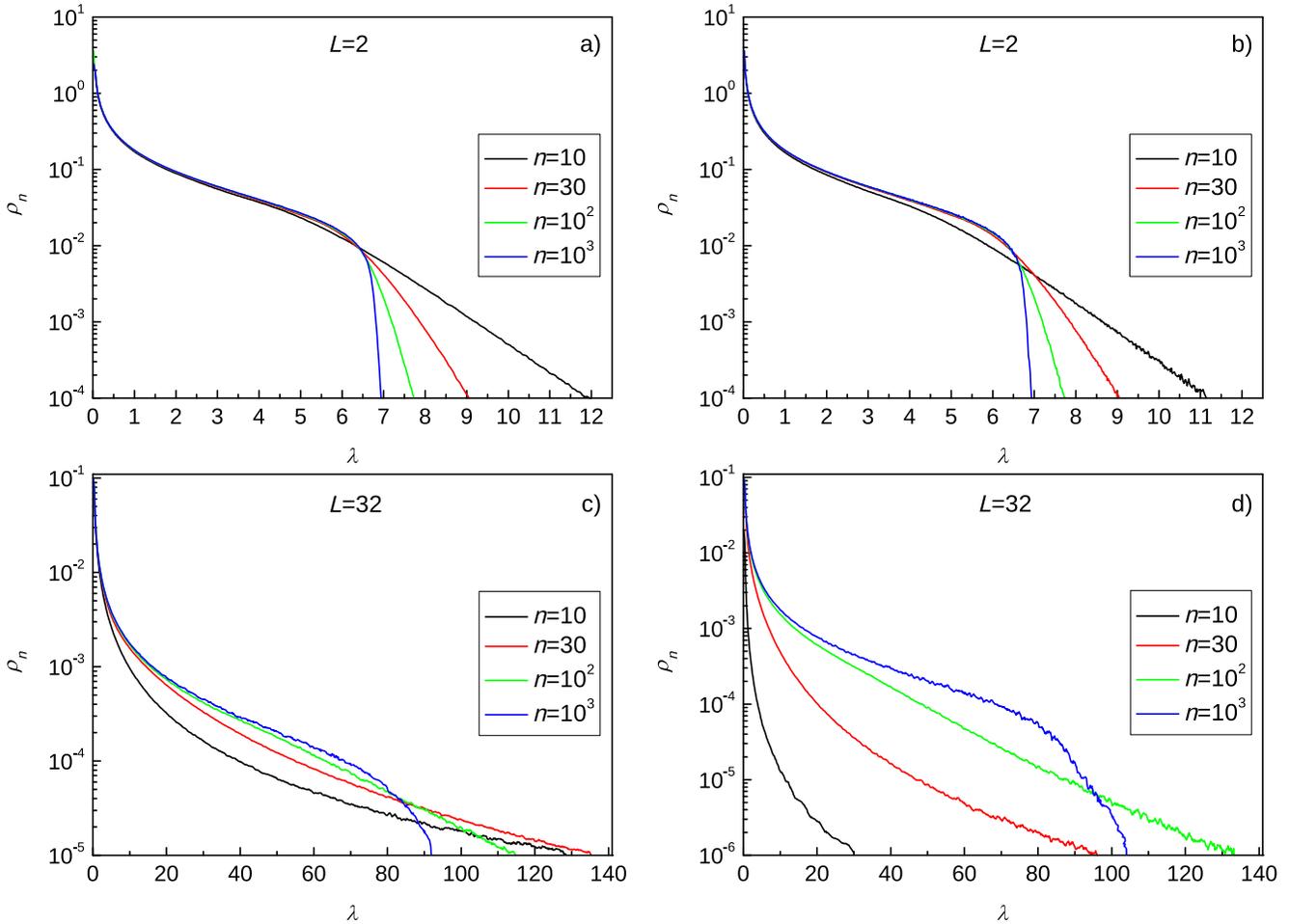} }
\caption{\small
%Singular values distribution $P(s)$ of (\ref{JJM}). Layer width $n = %10,30,100,1000$.
%The dependencies $P(s)$ were obtained by averaging over  $N=10^6$ realization %for $n=10$ and $n=30$, $N=10^5$ for $n=100$ and $N=10^4$ for $n=1000$.
%The matrix elements of weight matrices $W$ (\ref{wl}) and components of bias %vectors $b$ (\ref{bl}) have
%normal distribution with zero mean and $\sigma^2=1/n$.
%The figures a) and  b) correspond to linear activation function $\phi$, c) and %d) --- to ``Hard-Tanh'' activation function.
%In the figures a) and c) network depth $L=2$, in b) and d) network depth $L=32$.
The arithmetic means $\rho_{n}$ (in the semi-log scale) of the sample eigenvalue densities
of $M_{n}^{L}$ for various $L$, $n$ and $\varphi $
obtained by averaging over $N=10^{7}$ samples for $n=10,\;30$, $N=10^{6}$
samples for $n=10^{2}$ and $N=10^{4}$ samples for $n=10^{3}$.
Figures a) and c) correspond to linear activation function $\varphi $,
figures b) and d) correspond to the Hard-Tanh activation function
(see (\ref{hata})).
}
\label{fig3}
\end{figure}

Figure 3 shows (in the semi-log scale) the arithmetic means $\rho _{n}$ of
the sample eigenvalue densities  of $M_{n}^{L}$ (\ref{JJM}) with Gaussian weights and biases for
various $L$, $n$ and $\varphi $ obtained from $N$ samples chosen according to (\ref{Nn}). The "rows" of the figure, i.e., Fig. 3a) -- Fig. 3b) and Fig. 3c) -- 3d), describe the variation of $\rho _{n}$ in $n$ and $\varphi $ for a fixed $%
L=2,32$, while the "columns" of the figure, i.e., Fig. 3a) -- Fig. 3c) and Fig. 3b) -- 3d), describe
the variation of $\rho _{n}$ in $n\,$\ and $L$ for a fixed $\varphi
$, the linear or the HardTanh (\ref{hata}). We observe the
mentioned above similarity ("universality")  of curves
corresponding to different $\varphi ^{\prime }$, the stronger dependence of curves on $n$
and stronger fluctuations in $L$, especially near the upper edge $a_{L}$ of the support
and for the (non-smooth) HardTanh $\varphi$. It is also well seen the growth of $a_{L}$ in $L$.
 It is instructive to compare these properties of $\rho_n$ and those of the simplest case
 of $M_{n}^{L}$ with $%
R_{n}=D_{n}=\mathbf{1}_{n}$ (see (\ref{JJM})), where we have for the
infinite width limit of the Stieltjes transform $f_{L}$ and the eigenvalue
density $\rho_L $ (see \cite{Pa-Sh:11}, Problem 7.6.4):
%\begin{equation*}
  $z^{L}(-f_{L}(z))^{L+1}+zf_{L}(z)+1=0$
%\end{equation*}%
and
\begin{eqnarray}
\rho_L (\lambda ) &=&const.\ \lambda ^{-\alpha _{L}}(1+o(1)),\;\lambda
\downarrow 0,\;\;\alpha _{L}=L/(L+1), \notag \\
\rho_L (\lambda ) &=&const. \ (a_{L}-\lambda )^{1/2}(1+o(1)),\;\lambda \uparrow
a_{L},\; \notag \\
a_{L} &=&L(1+L^{-1})^{L+1}=Le(1+o(1)),\;L\rightarrow \infty. \label{rhos}
\end{eqnarray}%
It follows from the above that the support of $\rho_L $ grows in $L $
as well as its singularity at zero.
Moreover, it is easy to see that $\lim_{L\rightarrow \infty }f_L(z)=-z^{-1}$,
hence, the limiting $\rho_\infty $ is the Dirac delta at zero. The last property is
valid in general case of $M_{n}^{L}$ of (\ref{JJM}) and can be obtained either by the
free probability argument \cite%
{Pe-Co:18,Mu:02} or by using Theorems \ref{t:main} and \ref{t:ind}.
Note that the
subsequent limits $n \to \infty$ and then $L \to \infty$ can be viewed as an
implementation of the heuristic inequality $1 \ll L \ll n$. For another implementation
where $L \to \infty, \; n \to \infty, \; L=o(n)$ (the double scaling limit in the
terminology of statistical mechanics)  see \cite{Pe-Co:18}.

\begin{figure}[ht]
\center{\includegraphics[width=1.0\linewidth]{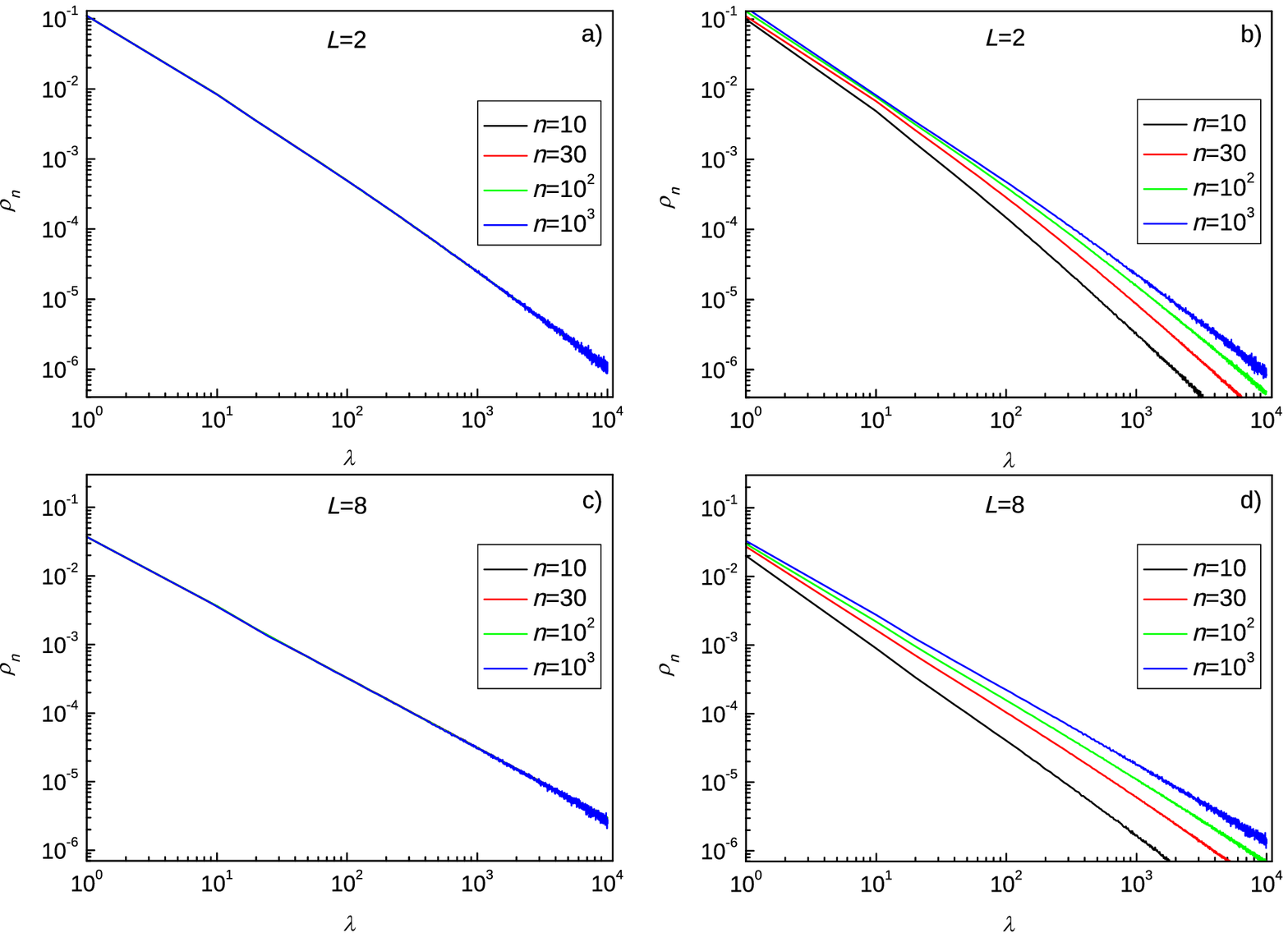} }
\caption{\small
The arithmetic means $\rho _{n}$ (in the double-log scale) of the sample eigenvalue
densities of $M_{n}^{L}$ with the Cauchy distributed weights and
biases (see (\ref{cauch})) for various $L$, $n$ and $\varphi $
obtained by averaging over $N=10^{7}$ samples for $n=10,\;30$,
$N=10^{6}$ samples for $n=10^{2}$ and $N=10^{4}$ samples for
$n=10^{3}$. Figures a) and c) correspond to linear activation
function $\varphi $, figures d) and d) correspond to the Hard-Tanh
activation function (see (\ref{hata})).
%The dependencies $P(s)$ were obtained by averaging over  $N=10^6$ realization %for $n=10$ and $n=30$, $N=10^5$ for $n=100$ and $N=10^4$ for $n=1000$.
%The matrix elements of weight matrices $W$ (\ref{wl}) and components of bias %vectors $b$ (\ref{bl}) have
%Cauchy distribution with location parameter $x_0=0$ and scale parameters $\gamma %= 1/n$.
%The figures a) and  b) correspond to linear activation function $\phi$, c) and %d) --- to ``Hard-Tanh'' activation function.
%In the figures a) and c) network depth $L=2$, in b) and d) network depth $L=8$.
%Upper panel corresponds to linear activation function, lower panel --- to %``Hard-Tanh'' activation function.
%Left figures correspond to network depth $L=2$, right figures --- to $L=8$.
%As seen in the pictures, for linear activation function $P(s)$ is well described %by power dependence $P(s)\sim s^{-\alpha}$ in wide range of $s$.
%The exponent $\alpha \approx 1.8$ for network depth $L=2$ and $\alpha \approx %1.2$  for $L=8$.
%More detailed analysis shows, that in the area of ``small'' $s$ ($0\leq %s %\lesssim 10$) the curve $P(s)$ bends and in the limit $s\to 0$ exponent
%$\alpha \approx 0.7$ for network depth $L=2$ and $\alpha \approx 0.9$ for %$L=8$.
}
\label{fig4}
\end{figure}

Figure 4 shows (in the double log-scale) the arithmetic means $\rho _{n}$ of
the sample eigenvalue densities  of $M_{n}^{L}$ (\ref{JJM}) with the Cauchy distributed (\ref{cauch}) weights and biases for
various $L$, $n$ and $\varphi $ obtained from $N$ samples chosen according to (\ref{Nn}). The figure is organized similarly to Figure 3, i.e., its "rows", Fig. 4a) -- Fig. 4b) and Fig. 4c) -- 4d), describe the variation of $\rho _{n}$ in $n$ and $\varphi $ for a fixed $%
L=2,8$, while the "columns", Fig. 4a) -- Fig. 4c)  and Fig. 4b) -- Fig. 4d), describe
the variation of $\rho _{n}$ $n\,$\ and $L$ for a fixed $\varphi
$, the linear or the HardTanh, see (\ref{hata}).

As seen from the pictures, for linear $\varphi$ the density $\rho_n$ is well described by the power law dependence $\rho_n(\lambda) \sim \lambda^{-\alpha}$ in a sufficient wide range of its argument with $\alpha \approx 1.25$ for network depth $L=2$ and $\alpha \approx 1.05$  for $L=8$.
More detailed analysis shows that for "small" $\lambda$ ($0\leq \lambda \lesssim 10$) the curve $\rho_n$ deviates from the straight line and can be described  for of sufficiently small $\lambda$'s by the power law with the exponent  $\alpha \approx 0.7$ for network depth $L=2$ and $\alpha \approx 0.9$ for $L=8$.

We observe also a certain similarity of Figure 4 and Figure 3, e.g., the stronger dependence of curves on $n$ and stronger fluctuations in $L$, especially for the case of non-smooth HardTanh $\varphi$ (not covered by Theorem \ref{t:main}).

Note that our analytic results do not apply to this case, since the Cauchy distribution does not satisfy conditions (\ref{bga}) - (\ref{wga}).
We present here these numerical results, firstly in order to demonstrate an example of a rather different behavior of the eigenvalue distribution density and, secondly, because of existing indications in the literature on the possibility of using random matrices with the "heavy-tailed" distributed entries  in the deep neural networks studies, see the review \cite{Ma-Ma:17} and references therein.

For more pictures of the eigenvalue distribution of $M_n^L$ and the related characteristics of the scheme (\ref{x0}) -- (\ref{bl}) see \cite{Pe-Co:18,Li-Qi:00,Ya:20} and references therein.

\section{Proof of Theorem 2.1.}

%Theorem \ref{t:main} of previous section is proved by induction in the
%layer number $l$, see formulas (\ref{m2m1}), (\ref{cm1m2}) and (\ref{mlml1}%
%). To pass from the $l$th to the $(l+1)$th layer we need a formula relating
%the limiting NCM $\nu _{\mathcal{M}^{l+1}}$ of the matrix $\mathcal{M}%
%_{n}^{l+1}$ and that of $\mathcal{M}_{n}^{l}$ in the infinite width limit $%
%n_{l}\rightarrow \infty $. The corresponding result, Theorem \ref{t:ind},
%which could be of independent interest, as well as certain auxiliary results
%are proved in this section. In particular, functional equations relating the
%Stieltjes transform of $\nu _{\mathcal{M}_{n}^{l+1}}$ and $\nu _{\mathcal{M}%
%_{n}^{l}}$ in the limit $n\rightarrow \infty $ are obtained (see (\ref{fhk})
%-- (\ref{kh})).

We begin with the list of facts of linear algebra and
probability theory that are used in the proof of Theorem \ref{t:ind}.
%and either are well known or can be easily proved.

\begin{proposition}
\label{p:facts} Let $A$ and $B$ be $n\times n$ real symmetric matrices and $L_{Y}$
be the rank one real symmetric matrix corresponding to the vector $Y=\{Y_{\alpha
}\}_{\alpha =1}^{n}\in \mathbb{C}^{n}$, i.e.,%
\begin{equation}
L_{Y}=\{L_{\alpha \beta }\}_{\alpha,\beta=1}^n, \;L_{\alpha \beta
}=Y_{\alpha }Y_{\beta }.  \label{L}
\end{equation}%
We have:

(i)
\begin{equation}
\mathrm{Tr}AL_{Y}=(AY,Y)  \label{tral}
\end{equation}

(ii) if%
\begin{equation}
G_{A}(z)=(A-z)^{-1},\;G_{B}(z)=(B-z)^{-1},\;\Im z\neq 0,  \label{ress}
\end{equation}%
are the resolvents of $A$ and $B$, then the resolvent identity%
\begin{equation}
G_{A}(z)=G_{B}(z)-G_{A}(z)(A-B)G_{B}(z),\;\Im z\neq 0  \label{resi}
\end{equation}%
is valid;

(iii) if $K$ is a real number and
\begin{equation}
A=B+KL_{Y},  \label{r11}
\end{equation}%
where $L_{Y}$ is given by (\ref{L}), then the rank one perturbation formula%
\begin{equation}
G_{A}(z)=G_{B}(z)-\frac{K}{1+K(G_{B}(z)Y,Y)}G_{B}(z)L_{Y}G_{B}(z),\;\Im
z\neq 0  \label{r12}
\end{equation}%
is valid, if $C$ is one more $n\times n$ matrix (not necessarily hermitian), then%
\begin{equation}
n^{-1}\mathrm{Tr}G_{A}(z)C-n^{  -1}\mathrm{Tr}G_{B}(z)C=-\frac{1}{n}\cdot \frac{%
K(G_{B}(z)CG_{B}(z)Y,Y)}{1+K(G_{B}(z)Y,Y)}  \label{tr12}
\end{equation}%
and if $B$ is positive definite and $K\geq 0$, then
\begin{equation}
|n^{-1}\mathrm{Tr}G_{A}(-\xi )C-n^{-1}\mathrm{Tr}G_{B}(-\xi )C|\leq
||C||/n\xi, \;\xi >0;  \label{troc}
\end{equation}%
(iv) if $X=\{X_{\alpha }\}_{\alpha =1}^{n}\in \mathbb{R}^{n}$ is a random
vector with jointly independent and identically distributed real components and (cf. (\ref%
{wga}) and (\ref{wm4}))%
\begin{equation}
\mathbf{E}\{X_{\alpha }\}=0,\;\mathbf{E}\{X_{\alpha }^{2}\}=1,\;\mathbf{E}%
\{X_{\alpha }^{4}\}=m_{4}<\infty ,  \label{m4}
\end{equation}%
where $\mathbf{E}\{\dots\}$ denotes the corresponding expectation, and
\begin{equation}
Y=n^{-1/2}SX,  \label{ysx}
\end{equation}%
with a $X$-independent $n\times n$ real symmetric matrix $S$, then
\begin{equation}
\mathbf{E}\{L_{Y}\}=n^{-1}R,\;R=S^{2}  \label{eys}
\end{equation}%
and if
\begin{equation}
a=(CY,Y),  \label{a}
\end{equation}%
is a quadratic form with a $X$-independent and not necessarily
hermitian $C$, then
\begin{equation}
\mathbf{E}\{a\}=n^{-1}\mathrm{Tr}\,CR=n^{-1}\mathrm{Tr}\,C_{S},\;C_{S}=SCS
\label{ea}
\end{equation}%
and%
%\begin{equation*}
%\mathbf{E}\{|a|^{2}\}=n^{-2}|\mathrm{Tr}\,C_{S}|^{2}+2n^{-2}\mathrm{Tr}\,%
%C_{S}C_{S}^{\ast }+(m_{4}-3)n^{-2}\sum_{\alpha =1}^{n}|(A_{S})_{\alpha
%\alpha }|^{2},
%\end{equation*}%
%, hence,%
\begin{equation}
\mathbf{Var}\{a\}:=\mathbf{E}\{|a|^{2}\}-|\mathbf{E}\{a\}|^{2}\leq \mu \;%
\mathrm{Tr}\,C_{S}C_{S}^{\ast }/n^{2},\;\mu =m_{4}+1,  \label{exa}
\end{equation}%
where $C^{\ast }$ is the hermitian conjugate of $C$ (recall that $ m_4 \in [1, \infty)$ in view of (\ref{m4})).
%Note than $m_{4}\geq 1$ in view of (\ref{m4}) (cf. (\ref{wga})).
\end{proposition}

\begin{proof}
Assertions (i) and (ii) are elementary.

(iii). To obtain (\ref{r12}) we use (\ref{resi}) with $A$ and $B$ of (\ref%
{r11}) to write the formula%
\begin{equation}
G_{A}(z)=G_{B}(z)-KG_{A}(z)LG_{B}(z).  \label{resi1}
\end{equation}%
Multiplying it by $L$ of (\ref{L}) from the right and using $%
LG_{B}(z)L=(G_{B}(z)Y,Y)L$, we get
\begin{equation}
G_{A}(z)L=(1+K(G_{B}(z)Y,Y))^{-1}G_{B}(z)L.  \label{resi2}
\end{equation}%
Plugging this into the r.h.s. of (\ref{resi1}), we obtain (\ref{r12}) and
then (\ref{tr12}).

Note that the r.h.s. of (\ref{r12}) and (\ref{tr12}) are well defined for $%
\Im z\neq 0$. Indeed, it follows from (\ref{resi}) with $A-z$ as $A$ and \ $%
A-z^{\ast }$as $B$\ that
\begin{equation*}
\Im G(z)=(2i)^{-1}(G-G^{\ast })=\Im z\ G^{\ast }(z)G(z),
\end{equation*}%
thus
\begin{equation*}
|1+K(G_{B}(z)Y,Y)|\geq |K\Im (G_{B}(z)Y,Y)|=|K|\ |\Im z|\
||G_{B}(z)Y||^{2}>0,\ \Im z\neq 0.
\end{equation*}%
To get (\ref{troc}) we take into account that $K\geq 0$ and $(G_{B}(-\xi
)Y,Y)\geq 0$, since $B$, hence, $G_{B}(-\xi )$ is positive definite. This yields
the following bound for the r.h.s. of (\ref{tr12})
\begin{equation*}
||C||\,||G_{B}(-\xi )Y)||^{2}/(G_{B}(-\xi )Y,Y)=||C||(G_{B}^{2}(-\xi
)Y,Y)/(G_{B}(-\xi )Y,Y)
\end{equation*}%
implying
\begin{align*}
& (G_{B}(-\xi )Y,Y)=(G_{B}^{-1}(-\xi )G_{B}^{2}(-\xi
)Y,Y)=((B+\xi )G_{B}^{2}(-\xi )Y,Y) \\
& \hspace{1.5cm}=(G_{B}(-\xi )BG_{B}(-\xi )Y,Y)+\xi (G_{B}^{2}(-\xi
)Y,Y)\geq \xi (G_{B}^{2}(-\xi )Y,Y).
\end{align*}%

(iv) We use the formulas (see (\ref{m4}))%
\begin{align*}
\mathbf{E}\{X_{\alpha _{1}}X_{\alpha _{2}}\} =\delta _{\alpha_{1}\alpha_{2}}, \\
\mathbf{E}\{X_{\alpha _{1}}X_{\alpha _{2}}X_{\alpha _{3}}X_{\alpha _{4}}\}
&=\delta _{\alpha_{1}\alpha_{2}}\delta _{\alpha_{3}\alpha_{4}}+\delta _{\alpha_{1}\alpha_{3}}\delta
_{\alpha_{2}\alpha_{4}}
\\&\hspace{0.3cm}+\delta _{\alpha_{1}\alpha_{4}}\delta _{\alpha_{2}\alpha_{3}}+(m_{4}-3)\delta
_{\alpha_{1}\alpha_{2}}\delta _{\alpha_{1}\alpha_{3}}\delta _{\alpha_{1}\alpha_{4}}.
\end{align*}%
The first line above, (\ref{L}) and (\ref{ysx}) yields (\ref{eys})\ and
(\ref{ea}), while the second line implies%
\begin{eqnarray*}
\mathbf{E}\{|a|^{2}\} &=&n^{-2}|\mathrm{Tr\,}C_{S}|^{2}+n^{-2}\sum_{\alpha
\neq \beta =1}^{n}|(C_{S})_{\alpha \beta }|^{2} \\
&&+n^{-2}\sum_{\alpha \neq \beta =1}^{n}(C_{S})_{\alpha \beta
}(C_{S})_{\beta \alpha }^{\ast }+(m_{4}-1)n^{-2}\sum_{\alpha
=1}^{n}|(C_{S})_{\alpha \alpha }|^{2}.
\end{eqnarray*}%
The sums of the r.h.s. are bounded by%
\begin{equation*}
n^{-2}\sum_{\alpha \neq \beta =1}^{n}|(C_{S})_{\alpha \beta }|^{2}=n^{-2}%
\mathrm{Tr\,}C_{S}C_{S}^{\ast }
\end{equation*}%
and $m_{4}\geq 1$ in view of (\ref{m4}). This leads (\ref{exa}).
\end{proof}
%add introd!!!
%%%%%%%%%%%%%%%%%%%%%%%%%%%%%%
%We will prove now the main result of this section.
%%%%%%%%%%%%%%%%%%%%%%%%proof t:indga

We will prove now Theorem \ref{t:ind}.

\begin{proof}
We begin with using Lemma \ref{l:mart} (i) below implying that the fluctuations (\ref{mnu4}) of $\nu _{%
\mathcal{M}_{n}}$ vanish sufficiently fast as $n\rightarrow \infty $. This
and the Borel-Cantelli lemma imply that
\[
\lim_{n \to \infty} |\nu_{\mathcal{M}_{n}}(\Delta)-\mathbf{E}\{\nu _{\mathcal{M}_{n}}(\Delta)\}|=0
\]
with probability 1, hence reduce the proof of the theorem to
the proof of the weak convergence of the expectation
\begin{equation}
\overline{\nu }_{\mathcal{M}_{n}}:=\mathbf{E}\{\nu _{\mathcal{M}_{n}}\}
\label{numc}
\end{equation}%
of $\nu _{\mathcal{M}_{n}}$ to the limit $\nu _{\mathcal{M}}$ whose
Stieltjes transform is given by (\ref{fhk}) -- (\ref{hcond}). Since $\mathcal{M}%
_{n}$ is positive definite, hence, its spectrum belongs to the closed positive
semiaxis $\mathbb{R}_{+}$ for all $n$, it suffices to prove the tightness of the
sequence of measures $\{\overline{\nu }_{\mathcal{M}_{n}}\}_{n}$ and the
pointwise convergence on a set of positive Lebesgue measure in $\mathbb{C}%
\setminus \mathbb{R}_{+}$ of their Stieltjes transforms (cf. (\ref{stm}))
\begin{equation}
f_{\mathcal{M}_{n}}(z):=\int_{0}^{\infty }\frac{\overline{\nu }_{\mathcal{M}%
_{n}}(d\lambda )}{\lambda -z}, \;  \mathbb{C}%
\setminus \mathbb{R}_{+} \label{fmnc}
\end{equation}%
to the limit satisfying (\ref{fhk}) -- (\ref{hcond}), see, e.g. \cite%
{Pa-Sh:11}, Proposition 2.1.2.

The tightness is guaranteed by the uniform in $n$ bound for%
\begin{equation}
\mu _{n}^{(1)}:=\int_{0}^{\infty }\lambda \overline{\nu }_{\mathcal{M}%
_{n}}(d\lambda ),  \label{mom}
\end{equation}%
since for any $T>0$ we have for the tail of $ \overline{\nu }_{\mathcal{M}%
_{n}}$
\[
\int_T^\infty \overline{\nu }_{\mathcal{M}%
_{n}}(d\lambda ) \le T^{-1}\int_{T}^{\infty }\lambda \overline{\nu }_{\mathcal{M}%
_{n}}(d\lambda )\le \mu _{n}^{(1)}/T.
\]
According to the definition of the NCM (see, e.g. (\ref{ncm})), spectral
theorem and (\ref{mncal}), we have
\begin{equation*}
\mu _{n}^{(1)}=\mathbf{E}\{n^{-1}\mathrm{Tr}\,\mathcal{M}_{n}\}=\mathbf{E}%
\{n^{-2}\mathrm{Tr}\,X_{n}R_{n}X_{n}^{T}K_{n}\}
\end{equation*}%
and then (\ref{tab}), (\ref{r2}) -- (\ref{Xn}) and (\ref{Dn}) -- (\ref{Phi})
yield%
\begin{equation}
\mu _{n}^{(1)}\leq n^{-2}\Phi _{1}^{2}\mathbf{E}\{\mathrm{Tr}\,%
X_{n}R_{n}X_{n}^{T}\}=\Phi _{1}^{2}n^{-1}\mathrm{Tr}\,R_{n}\leq
r_{2}^{1/2}\Phi _{1}^{2},  \label{tim}
\end{equation}%
where we used the inequality $n^{-1}\mathrm{Tr}\,R_{n} \le (n^{-1}\mathrm{Tr}\,R_{n}^2)^{1/2}$ to obtain the r.h.s. bound. This implies the tightness of $\{\overline{\nu }_{\mathcal{M}_{n}}\}_{n}$
and reduces the proof of the theorem to the proof of pointwise in $\mathbb{C}%
\setminus \mathbb{R}_{+}$ convergence of (\ref{fmnc}) to the limit
determined by (\ref{fhk}) -- (\ref{kh}).

The above argument, reducing the analysis of the large size behavior of the
eigenvalue distribution of random matrices to that of the expectation of the Stieltjes transform of the distribution, is widely used in random matrix
theory (see \cite{Pa-Sh:11}, Chapters 3, 7, 18 and 19), in particular, while
dealing with the sample covariance matrices. However, the matrix $\mathcal{M}%
_{n}$ of (\ref{mncal}) differs essentially from the sample covariance
matrices, since the "central" matrix $K_{n}$ of (\ref{Dn}) is random and
dependent on $X_{n}$ (the data matrix according to statistics), while in the
sample covariance matrix the analog of $K_{n}$ is either deterministic or
random but independent of $X_{n}$. Nevertheless, we show that in our case of
the $X_{n}$-dependent $K_{n}$ of (\ref{Dn}) it suffices to follow
essentially the proof for $X_{n}$-independent analogs of $K_{n}$, which
dates back to \cite{Ma-Pa:67,Pa:72} and has been largely extended and
used afterwards, see, e.g. \cite{Pa-Sh:11,Ak-Co:11,Ba-Si:00,Gi:12}.

We outline first the scheme of the proof of the theorem. Write (\ref{mncal}) as
\begin{equation}
\mathcal{M}_{n}=\sum_{j=1}^{n}K_{jn}L_{jn},  \label{mnr1}
\end{equation}%
where $K_{jn}$ are given by (\ref{Dn}) and
\begin{equation}
L_{jn}=L_{Y_{j}}  \label{lljn}
\end{equation}%
is the rank-one matrix (\ref{L}) corresponding to the random vector (cf. (%
\ref{ysx}))
\begin{equation}
Y_{j}=\{n^{-1/2}(S_{n}X_{n}^{T})_{j\alpha }\}_{\alpha
=1}^{n}=n^{-1/2}S_{n}X_{j},\;X_{j}=\{X_{j\alpha }\}_{\alpha =1}^{n},
\label{Yj}
\end{equation}%
i.e., $X_{j}$ is the $j$th row of the random matrix (\ref{Xn}), thus the
collection $\{X_{j}\}_{j=1}^{n}$ consists of i.i.d. random vectors
satisfying (\ref{m4}).

It follows from the definition of the Normalized Counting Measure (see (\ref%
{ncm})) and the spectral theorem for the resolvent%
\begin{equation}
G_{\mathcal{M}_{n}}(z)=(\mathcal{M}_{n}-z)^{-1}  \label{gmnc}
\end{equation}%
of $\mathcal{M}_{n}$ that
\begin{equation}
f_{\mathcal{M}_{n}}(z)=\mathbf{E}\{n^{-1}\mathrm{Tr}\,G_{\mathcal{M}_{n}}(z)\}.
\label{fntr}
\end{equation}%
Hence, we have to deal with $G_{\mathcal{M}_{n}}(z)$.

By using the resolvent
identity (\ref{resi}) for $A=\mathcal{M}_{n}$ and $B=0$, we obtain for (\ref%
{gmnc}) in view of (\ref{mnr1})%
\begin{equation}
G(z)=-z^{-1}+z^{-1}\sum_{j=1}^{n}K_{j}G(z)L_{j},  \label{cg1}
\end{equation}%
and we omit here and below the subindex $\mathcal{M}_{n}$ in the resolvent (%
\ref{gmnc}) as well as the subindex $n$ %and the argument $z$
in many instances below where this does no lead to confusion.

Next, we choose in (\ref{r12})%
\begin{equation}
A=\mathcal{M}_{n},\;B=\mathcal{M}_{n}^{(j)}:=\mathcal{M}_{n}-K_{j}L_{j}
\label{mmj}
\end{equation}%
and use (\ref{resi2})
%multiply it by $L_j$ from the right
to obtain
\begin{equation}
G(z)L_j=G_{j}(z)L_j(1+K_{j}a_{j}(z))^{-1},\;a_{j}(z):=(G_{j}(z)Y_{j},Y_{j}),
\label{ggj}
\end{equation}%
where
\begin{equation}
G_{j}(z):=(\mathcal{M}_{n}^{(j)}-z)^{-1}  \label{cgj}
\end{equation}%
and $Y_{j}$ is defined in (\ref{Yj}).
Plugging (\ref{ggj}) into (\ref{cg1}), we obtain our basic starting formula%
\begin{equation}
G(z)=-z^{-1}+z^{-1}\sum_{j=1}^{n}\frac{K_{j}}{(1+K_{j}a_{j}(z))}%
G_{j}(z)L_{j}  \label{cg20}
\end{equation}
which we are going to convert into the "prelimit" version of the
system (\ref{fhk}) -- (\ref{kh}) %, i.e., a form expressed via $%
%G(z)$ but not $G_{j}(z)$ and $L_{j}$ (note that $a_{j}(z)=\mathrm{Tr}\,%
%G_{j}(z)L_{j}$ by (\ref{tral}))
plus error terms vanishing as $n\rightarrow \infty $.

It follows from (\ref{fntr}) that we are allowed to make any modification of
(\ref{cg20}) provided that the corresponding error term $\mathcal{E}_n$
satisfies
\begin{equation}
\mathbf{E}\{n^{-1}\mathrm{Tr}\,\mathcal{E}_{n}\}=o(1),\;n\rightarrow \infty .
\label{err}
\end{equation}%
Denote $\mathbf{E}_{j}\{...\}$ the operation of  expectation conditioned on $\{X_k\}_{k \neq j}$
%only with respect to $X_{j}$
and use:

(i) (\ref{m4}) -- (\ref{exa}) and (\ref{Yj}) to
replace the random quadratic form  $(C Y_{j},Y_{j})$
%(\ref{ggj})
with a $Y_{j}$%
-independent matrix $C_{}$ by
\begin{equation}
\mathbf{E}_{j}\{(CY_{j},Y_{j})\}=n^{-1}\mathbf{E}_{j}\{(C_{S}X_{j},X_{j})%
\}=n^{-1}\mathrm{Tr}\,CR;  \label{exfj}
\end{equation}

(ii) (\ref{tr12}) -- (\ref{troc}) to replace
$n^{-1}\mathrm{Tr}G_{j}(z)C$ with a $Y_{j}$%
-independent $C$ by $n^{-1}\mathrm{Tr}\,G(z)C$;
%with the
%note the factor $1/n$ in the r.h.s. of the formula;

(iii) Lemma \ref{l:mart} to replace the random variable $n^{-1}\mathrm{Tr}\,%
G(z)C$ with a $Y_{j}$%
-independent matrix $C$ by the expectation $\mathbf{E}\{n^{-1}%
\mathrm{Tr}\,G(z)C\}$.

We will  apply then:  (i) with $C=G_{j}$ to replace $a_{j}(z)$ of (\ref{ggj})
by
\begin{equation}
h_{jn}(z)=n^{-1}\mathrm{Tr}\,RG_{j}(z),  \label{hj}
\end{equation}%
 (ii) with $C=R$ to replace $h_{jn}(z)$ by
$h_{n}(z)$ and then
(iii) with $C=R$ to replace $h_{n}(z)$ by $\overline{h_{n}}(z)$, where%
\begin{equation}
h_{n}(z)=n^{-1}\mathrm{Tr}\,RG(z)=n^{-1}\mathrm{Tr}\,SG(z)S,\;\; \overline{h_{n}}%
(z)=\mathbf{E}\{h_{n}(z)\}.  \label{hns}
\end{equation}%
As a result, we can replace $(1+K_{j}a_{j}(z))$ by $(1+K_{j}\overline{h_{n}}%
(z))$ in the r.h.s. of (\ref{cg20}), see Lemma \ref{l:dif} for details.

Likewise, we can replace $L_{j}$ in (\ref%
{cg20}) by its expectation $n^{-1}R$ by using (i)
and then replace $G_{j}(z)$ by $G(z)$ by using (ii) to convert (\ref%
{cg20}) into%
\begin{equation}
G(z)=-z^{-1}+z^{-1}\overline{k}_{n}(z)G(z)R+T_{1}(z),  \label{cg30}
\end{equation}%
where%
\begin{equation}
k_{n}(z)=\frac{1}{n}\sum_{j=1}^{n}\frac{K_{j}}{1+K_{j}\overline{h}_{n}(z)},\;%
\overline{k}_{n}(z)=\mathbf{E}\{k_{n}(z)\}  \label{kns}
\end{equation}%
and%
\begin{equation}
T_{1}(z)=z^{-1}\sum_{j=1}^{n}\frac{K_{j}}{1+K_{j}a_{j}(z)}%
G_{j}(z)L_{j}-z^{-1}\overline{k}_{n}(z)G(z)R.  \label{ct10}
\end{equation}
is the error term.

Applying to (\ref{cg30}) the operation $\mathbf{E}\{n^{-1}\mathrm{Tr\;}\dots\}$
and taking into account (\ref{fntr}) and (\ref{hns}), we get%
\begin{equation}
f_{\mathcal{M}_{n}}(z)=-z^{-1}+z^{-1}\overline{k}_{n}(z)\overline{h}%
_{n}(z)+t_{1n},  \label{fhkn}
\end{equation}%
where%
\begin{equation}
t_{1n}(z)=\mathbf{E}\{n^{-1}\mathrm{Tr}\, T_{1}(z)\},  \label{t10}
\end{equation}%
i.e., a "prelimit" version of (\ref{fhk}) with the error term $t_{1n}$, cf. (%
\ref{err}).

Next, we have from (\ref{cg30})%
\begin{equation}
G(z)=\mathcal{G}(z)+zT_{1}\mathcal{G}(z),\;\mathcal{G}(z)=(\overline{k}%
_{n}(z)R-z)^{-1}.  \label{cg40}
\end{equation}%
Multiplying the formula by $R$, applying to the result the operation $%
\mathbf{E}\{n^{-1}\mathrm{Tr\;}...\}$ and using the fact that $R$ of (\ref%
{r2}), hence, $\mathcal{G}$ are independent of $X_n$ of (\ref{Xn}), we obtain in view of (\ref{hns})%
\begin{equation}
\overline{h}_{n}(z)=\int_{0}^{\infty }\frac{\lambda \nu _{R_{n}}(d\lambda )}{%
\lambda \overline{k}_{n}(z)-z}+t_{2n},  \label{hkn}
\end{equation}%
where $\nu _{R_{n}}$ is the Normalized Counting Measure of $R_{n}$ defined
in (\ref{nur}). This is a "prelimit" version of (\ref{hk}) with the error
term (cf. (\ref{err}) and (\ref{t10}))
\begin{equation}
t_{2n}(z)=-z\mathbf{E}\{n^{-1}\mathrm{Tr}\, T_{1}\mathcal{G}(z)R\}.  \label{t20}
\end{equation}%
At last, observing that according to the conditions of the theorem (see (\ref%
{Xn}), (\ref{bl}) and (\ref{Dn})) $\{K_{j}\}_{j=1}^{n}$ are independent
identically distributed (and, possibly, $n$-dependent) random variables, we
obtain from (\ref{kns}) the "prelimit version"%
\begin{equation}
\overline{k}_{n}(z)=\int_{0}^{\infty }\frac{\lambda \overline{\nu }%
_{K_{n}}(d\lambda )}{\lambda \overline{h}_{n}(z)+1}  \label{khn}
\end{equation}%
of (\ref{kh}) in which $\overline{\nu }_{K_{n}}$ is the probability law of
(see (\ref{Dn}))
\begin{equation}
K_{jn}=(\varphi ^{\prime }(\eta _{jn}+b_{j}))^{2},\;\eta
_{jn}=n^{-1/2}\sum_{\alpha =1}^{n}X_{j\alpha }x_{\alpha n}.  \label{nukn0}
\end{equation}%
Having obtained semi-heuristically relations (\ref{fhkn}), (\ref{hkn}) and (%
\ref{khn}), we pass now to their rigorous derivation, i.e., to the proof
that the remainders $t_{1n}$ of (\ref{t10}) and $t_{2n}$ of (\ref{t20})
vanish in the limit $n\rightarrow \infty $ and that $f_{\mathcal{M}_{n}}$, $%
\overline{h}_{n}$ and $\overline{k}_{n}$ converge to a solution of (\ref{fhk}%
) -- (\ref{kh}).

We will deal first with the $n \to \infty$ limit in (\ref{fhkn}), (\ref{khn})
and (\ref{hkn}) assuming that $t_{1n}$ and $t_{2n}$ vanish as $n\rightarrow
\infty $. In fact, the limit is a version of that widely used in random
matrix theory, see, e.g. \cite{Pa-Sh:11}. Thus, we just outline the
procedure.

According to (\ref{fmnc}) $f_{\mathcal{M}_{n}}$ is analytic in $\mathbb{C}%
\setminus \mathbb{R}_{+}$ for every $n$. Thus, by Vitali's theorem on the convergence of analytic functions, it suffices to study the limiting properties of the sequence $\{f_{\mathcal{M}_{n}}\}_{n}$ for $z$
varying in a closed interval of the open negative semiaxis%
\begin{equation}
I_{-}=\{z\in \mathbb{C}:z=-\xi ,\;0<\xi _{-}\leq \xi \leq \xi _{+}<\infty \},
\label{imi}
\end{equation}%
where $\xi_\pm$ do not depend on $n$.

Furthermore, since $\mathcal{M}_{n}$ of (\ref{mncal}) and $\mathcal{M}%
_{n}^{(j)}$ of (\ref{mmj}) are positive definite, their resolvents $G(z)$
and $G_{j}(z )$ for $z=-\xi \in I_{-}$ are also positive definite, thus
\begin{equation}
||G(-\xi )||\leq 1/\xi ,\;||G_{j}(-\xi )||\leq 1/\xi,  \; \;\xi >0,  \label{nors}
\end{equation}
and we have for $a_{j}$ of (\ref{ggj})
\begin{equation}
a_{j}(-\xi )\geq 0,\;\xi >0.  \label{aj}
\end{equation}%
Besides, we have from (\ref{Dn}) and (\ref{Phi})%
\begin{equation}
0\leq K_{j}\leq \Phi _{1}^{2}.  \label{Kj}
\end{equation}%
Thus, $1+K_{j}a_{j}(-\xi )\geq 1$ and (\ref{ggj}) is well defined for $\xi
>0$.

It follows from  (\ref{hns}),  spectral theorem for $\mathcal{M}_{n}$ and (\ref{r2}) that
if $\{\lambda_\alpha\}_\alpha$ and $\{\psi_\alpha\}_\alpha$ are the eigenvalues and the eigenvectors of $\mathcal{M}_n$, then
\begin{align}
&\hspace{0cm} h_n(z)=\int_0^\infty \frac{\mu_n (d \lambda)}{\lambda-z}, \; z \in \mathbb{C} \setminus \mathbb{R}_+,
\notag
\\&\hspace{0.5cm}  \mu_n =n^{-1} \sum_{\alpha} \delta_{\lambda_\alpha} (R\psi_\alpha,\psi_\alpha), \;  \notag
\\&\hspace{-1cm} 0<\mu_n (\mathbb{R}_+)=n^{-1} \sum_{\alpha}(R\psi_\alpha,\psi_\alpha)=n^{-1}\mathrm{Tr}\, R \le r_2^{1/2},\label{hrep}
\end{align}
where we took into account that $R$ is positive defined and used Schwarz inequality for traces. This implies that $h_n$ is analytic in $\mathbb{C} \setminus \mathbb{R}_+$ and
\begin{equation}\label{hnev}
\Im h_n(z)  \Im z =\int_0^\infty \frac{\mu_n (d \lambda)}{|\lambda-z|^2}>0, \; \Im z \neq 0.
\end{equation}
In particular, the function $\overline{k}_n$ of (\ref{kns}) is also analytic in $\mathbb{C} \setminus \mathbb{R}_+$.
%Since $\mathcal{M}_{n}$, $\mathcal{M}_{n}^{(j)}$
%since $S_{n}$ and $R_{n}=S_{n}^{2}$ of (\ref{r2}) are also positive
%definite, $SG(-\xi )S$ is positive definite as well for $\xi \in \lbrack \xi
%_{-},\xi _{+}]$ %by the spectral theorem, we have%
%\begin{equation}
%||G(-\xi )||\leq 1/\xi ,\;||G_{j}(-\xi )||\leq 1/\xi ,  \label{nors}%
%\end{equation}%

It follows also from the above and (\ref{fmnc}) that
%, (\ref{hns}) and (\ref{kns})
%for $z=-\xi \in I_{-}$ of (\ref{imi})
\begin{equation}
0<f_{\mathcal{M}_{n}}(-\xi )\leq 1/\xi ,\;\;0<\overline{h}_{n}(-\xi )\leq
r_{2}^{1/2}/\xi ,\;\; 0<\overline{k}_{n}(-\xi )\leq 1,\;\;\xi >0.
\label{pos}
\end{equation}%
Moreover, since the sequences $\{f_{\mathcal{M}_{n}}\}_n, \; \{\overline{h}_{n}\}_n$
and $\{\overline{k}_{n}\}_n$ are real analytic on $I_-$ of (\ref{imi}), there exists
a subsequence $n_{j}\rightarrow \infty $ such that $\{f_{\mathcal{M}%
_{n_{j}}}\},\{\overline{h}_{n_{j}}\}$ and $\{\overline{k}_{n_{j}}\}$
converge uniformly on (\ref{imi}) to certain limits $f_{\mathcal{M}},h$ and $%
k$ analytic in $\mathbb{C}\setminus \mathbb{R}_{+}$. This allows us to carry
out the limit along $n_{j}\rightarrow \infty $ in the second term in the
r.h.s. of (\ref{fhkn}) and to obtain (\ref{fhk}) provided that $t_{1n_{j}}$
vanishes as $n_{j}\rightarrow \infty $.

Next, write the first term in the r.h.s. of (\ref{hkn}) for $z=-\xi \in
I_{-} $ as%
\begin{equation*}
\int_{0}^{\infty }\frac{\lambda \nu _{R_{n_{j}}}(d\lambda )}{\lambda k(-\xi
)+\xi }+(\overline{k}_{n_{j}}(-\xi )-k(-\xi ))\int_{0}^{\infty }\frac{%
\lambda ^{2}\nu _{R_{n_{j}}}(d\lambda )}{(\lambda k(-\xi )+\xi )(\lambda
\overline{k}_{n_{j}}(-\xi )+\xi )}.
\end{equation*}%
It follows then from (\ref{r2}), (\ref{nur}) and (\ref{pos}) that the first
term tends to the r.h.s. of (\ref{hk}) as $n_{j}\rightarrow \infty $. The
integral in the second term is bounded by $r_{2}/\xi ^{2}$ in view of (\ref%
{r2}) and (\ref{pos}), hence, the second term vanishes as $n_{j}\rightarrow
\infty $. Thus, we obtain (\ref{hk}) provided that $t_{2n_{j}}$ vanishes as $%
n_{j}\rightarrow \infty $.

An analogous argument applies to (\ref{khn}). However, to obtain (\ref{kh})
and (\ref{nuka}), we have to find the limiting probability law of the random variable
$\eta _{jn}$ of (\ref{nukn0}). It follows from the standard facts on the
Central Limit Theorem (see, e.g. \cite{Sh:96}, Section III.4), (\ref{m4})
and (\ref{qqn}) that the law is Gaussian of zero mean and variance $q-\sigma
_{b}^{2}$, see Lemma \ref{l:clt} for details. This proves (\ref{kh}).

Thus, we have proved the validity of (\ref{fhk}) -- (\ref{nuka}) for $z\in
I_{-}$ of (\ref{imi}) provided that $t_{1n}$ and $t_{2n}$ of (\ref{t10}) and
(\ref{t20}) vanish uniformly in $z\in I_{-}$. It is shown in Lemma \ref%
{l:uniq} that the system (\ref{hk}) -- (\ref{kh}) is well defined and
uniquely solvable everywhere in $\mathbb{C}\setminus \mathbb{R}_{+}$. This
implies that the whole sequences $\{f_{\mathcal{M}_{n}}\},\{\overline{h}%
_{n}\}$ and $\{\overline{k}_{n}\}$ converge uniformly on any compact set of $%
\mathbb{C}\setminus \mathbb{R}_{+}$, that their limits $f_{\mathcal{M}},h$
and $k$ are not identically zero and can be found from relations (\ref{fhk})
-- (\ref{nuka}) which are valid everywhere in $\mathbb{C}\setminus \mathbb{R}%
_{+}$. Indeed, if $h$ is identically zero, then it follows from (\ref{kh})
with $z=-\xi <0$ that $k(-\xi )=k_{1}>0\,$, where $k_{1}$ is the first
moment of measure $\nu _{K}$ of (\ref{nuka}). Then (\ref{hk}) implies that $%
\nu _{R}$ is concentrated at zero. This contradicts condition (a) of the
theorem. Analogously, assuming that $k$ is identically zero, we conclude
that $\nu _{K}$ of (\ref{nuka}) is concentrated at zero and this is
impossible if $\varphi $ is not identically constant.
%Let us show that (\ref{hk}) and (\ref%
%{kh}) are valid everywhere in $\mathbb{C}\setminus \mathbb{R}_{+}$. It was
%proved above that the relations are valid for $z=-\xi <0$. \ ???

Besides, it follows from  (\ref{pos})
%(\ref{hrep}), (\ref{hk}) -- (\ref{kh}) for $z=-\xi <0$
that $k(-\xi )$ and $\xi h(-\xi )$ are nonnegative and bounded. %
Thus,
the limit $\xi \rightarrow \infty $ in (\ref{fhk}) yields (\ref{numr1}). Conditions (\ref{hcond}) follow from the $n \to \infty$ versions of
(\ref{hrep}) -- (\ref{pos}).

\smallskip
We pass now to the most technical part of the proof in which we establish
that the  error terms (\ref{t10}) and (\ref{t20}) vanish as $n \to \infty$ uniformly on $%
z\in I_{- }$ of (\ref{imi}).

Note first that it suffices to assume that the sequence $\{R_{n}\}$ of (\ref%
{mncal}) -- (\ref{nur}) is uniformly bounded, i.e.,
\begin{equation}
\sup_{n}||R_{n}||\leq \rho <\infty ,  \label{rrho}
\end{equation}%
instead of (\ref{r2}). This is also a standard and technically convenient
trick of random matrix theory where it is shown that once the limiting
Normalized Counting Measure is found under condition (\ref{rrho}), it can be
also found under condition (\ref{r2}), see e.g. \cite{Pa-Sh:11}, Section 19,
in particular, Theorem 19.1.8 for the case where $K_{n}$ is independent of $%
X_{n}$ of (\ref{Xn}). In our case of (\ref{Dn}) -- (\ref{Phi}) the proof of this fact is
given in \cite{Pa:20}.

By using (\ref{ct10}) and (\ref{t10}), we have from (\ref{tral}), (\ref{hns}%
) and (\ref{kns})%
\begin{align}
t_{1n}(-\xi )& =\frac{1}{zn}\sum_{j=1}^{n}\mathbf{E}%
\{a_{j}K_{j}(1+K_{j}a_{j})^{-1}-\overline{h}_{n}K_{j}(1+K_{j}\overline{h}%
_{n})^{-1}\}|_{z=-\xi }  \notag \\
& =-\frac{1}{zn}\sum_{j=1}^{n}\mathbf{E}\{(a_{j}-\overline{h}%
_{n})K_{j}((1+K_{j}a_{j})(1+K_{j}\overline{h}_{n}))^{-1}\}|_{z=-\xi }.
\label{t11}
\end{align}%
It follows then from (\ref{aj}) -- (\ref{pos}) that
\begin{equation}
(1+K_{j}a_{j}(-\xi ))\geq 1,\;(1+K_{j}\overline{h}_{n}(-\xi ))\geq 1.
\label{den}
\end{equation}%
These bounds, (\ref{pos}) and (\ref{t11}) imply%
\begin{equation}
|t_{1n}(-\xi )|\leq \Phi
_{1}^{2}d_{1}\xi^{-1},\;d_{1}=n^{-1}\sum_{j=1}^{n}d_{1j},\;d_{1j}=\mathbf{E}\{|a_{j}-%
\overline{h}_{n}|\}.  \label{t12}
\end{equation}%
According to Lemma \ref{l:dif}, $d_{1j}\leq C^{\prime }/n^{1/2}$ if $n$ is
large enough and we obtain%\
\begin{equation}
|t_{1n}(-\xi )|\leq C_{1}/n^{1/2},\;C_{1}=\Phi _{1}^{4}C^{\prime } \xi^{-1}.
\label{toc1}
\end{equation}%
This and (\ref{fhkn}) justify (\ref{fhk}).

Consider now $t_{2n}$ of (\ref{t20}). Using an argument similar to that
leading to (\ref{t11}), we obtain
\begin{align}
t_{2n}(-\xi )& =n^{-1}\sum_{j=1}^{n}\mathbf{E}%
\{b_{j}K_{j}(1+K_{j}a_{j})^{-1}-\overline{c}_{n}K_{j}(1+K_{j}\overline{h}%
_{n})^{-1}\}|_{z=-\xi }  \notag \\
& \hspace{2cm}=t_{2n}'(-\xi )+t_{2n}''(-\xi ),  \label{t21}
\end{align}%
where (cf. (\ref{ggj}) and (\ref{hns}))
\begin{equation}
b_{j}=(R\mathcal{G}G_{j}Y_{j},Y_{j}),\;c_{n}=n^{-1}\mathrm{Tr}\,R^{2}\mathcal{G%
}G,\;\overline{c}_{n}=\mathbf{E}\{c_{n}\}  \label{bp}
\end{equation}%
and (cf. \ref{t11})
\begin{align}
t_{2n}'(-\xi )& =n^{-1}\sum_{j=1}^{n}\mathbf{E}\{(\overline{h}%
_{n}-a_{j})K_{j}^{2}\overline{c}_{n}((1+K_{j}a_{j})(1+K_{j}\overline{h}%
_{n}))^{-1}\}|_{z=-\xi },  \notag \\
t_{2n}''(-\xi )& =n^{-1}\sum_{j=1}^{n}\mathbf{E}\{(b_{j}-\overline{c}%
_{n})(K_{j}+K_{j}^{2}\overline{h}_{n})((1+K_{j}a_{j})(1+K_{j}\overline{h}%
_{n}))^{-1}\}|_{z=-\xi }.
\label{t2122}
\end{align}%
Since $R$ is positive definite, (\ref{pos}) implies for $\mathcal{G}$ of (%
\ref{cg40}) (cf. (\ref{nors}))%
\begin{equation}
||\mathcal{G}(-\xi )||\leq \xi ^{-1}.  \label{norH}
\end{equation}%
It follows then from (\ref{nors}), (\ref{rrho}), (\ref{bp}) and (\ref{norH})
that\begin{equation}
|\overline{c}_{n}|\leq \rho ^{2}\xi ^{-2}.  \label{cbar}
\end{equation}%
This, (\ref{Kj}) and (\ref{den}) yield%
\begin{equation}
|t_{2n}'(-\xi )|\leq \Phi _{1}^{4}\rho ^{2}\xi ^{-2}d_{1}  \label{t21x}
\end{equation}%
with $d_{1}$ of (\ref{t12}). Thus, Lemma \ref{l:dif} implies%
\begin{equation}
|t_{2n}'(-\xi )|\leq C_{21}/n^{1/2}  \label{toc21}
\end{equation}%
for a certain $n$-independent $C_{21}$.

Likewise, by using (\ref{Kj}), (\ref%
{den}) and (\ref{pos}), we obtain%
\begin{equation}
|t_{2n}''(-\xi )|\leq (\Phi _{1}^{2}+\Phi _{1}^{4}r^{1/2}\xi
^{-1})\ d_{2},\; \ d_{2}=n^{-1}\sum_{j=1}^{n}d_{2j},\;d_{2j}=\mathbf{E}\{|b_{j}-%
\overline{c}_{n}|\},  \label{t211}
\end{equation}%
and then Lemma \ref{l:dif} implies $|t_{2n}''(-\xi )|\leq C_{22}/n^{1/2}$ for
a certain $n$-independent $C_{22}$.

Combining this bound, (\ref{t21}) and (%
\ref{toc21}), we get (cf. (\ref{toc1}))
\begin{equation}
|t_{2n}(-\xi )|\leq C_{2}/n^{1/2}.  \label{toc2}
\end{equation}%
This and (\ref{hkn}) justifies (\ref{hk}).
\end{proof}
%%%%%%%%%%%%%%%%%%%%%%%%%%%%%%%%%%%%%%%%%%%%%%%%%%%%%%%
\begin{remark}\label{r:kgen}
It is noted at the beginning of Section 2 that despite the
fact that the matrices $D^{l}$ of (\ref{D}), hence $K_{n}^{l}$ of (\ref{kan}%
), are random and depend on $X^{l}$ of (\ref{wga}), the limiting
eigenvalue distribution of $M_{n}^{L}$ of (\ref{JJM}) corresponds to the
case where $D^{l}$ of (\ref{D}) and $K_{n}^{l}$ are random but \emph{%
independent} of $X^{l}$, see (\ref{nukal}) and (\ref{nuka}). The emergence
of this remarkable property of $M_{n}^{L}$ is well seen in the above proof,
in particular, in formulas (\ref{cg30}) -- (\ref{nukn0}) and (\ref{toc1}), (%
\ref{toc2}). Moreover, it follows from the above proof that a quite general
dependence of $D^{l}$ on $X^{l}$ is possible provided that probability law
of the entries $\{K_{jn}\}_{j=1}^n$ of $K_{n}$ in (\ref{Dn}) are independent and their probability law admits a limiting form as $n\rightarrow \infty $. For
instance, we can replace $\{X_{ja}\}_{j=1}^{n}$ in $\eta _{jn}$ of (\ref%
{nukn0}) by, say,  $\{X_{ja}^{p}\}_{j=1}^{n}$ with a certain $p$.

It is also noteworthy that formulas (\ref{mnr1}) -- (\ref{Yj}) present the matrix $\mathcal{M}_n$ as the sum of \emph{jointly independent} rank 1 matrices. This, basic for the proof of the theorem (see also Lemma \ref{l:mart}), representation is the reason to pass from matrices $M_n^l$ of (\ref{JJM}) (see also (\ref{m1}) and (\ref{m2m1}))
to matrices $\mathcal{M}_n^l$, see(\ref{cm1}) (\ref{cm1m2}) and (\ref{mncal}).
The representation dates back to works \cite{Ma-Pa:67,Pa:72} and has being widely using since then in random matrix theory.
\end{remark}
%%%%%%%%%%%%%%%%%%%%%%%%%%%%%%%%%%%%%%%%%%%%%%%%%%%%%%%%%%%%%%%%%%
\begin{lemma}
\label{l:dif} Let $d_{1j}(-\xi)$ and $d_{2j}(-\xi)$ be defined in (\ref{t12}) and (\ref{t211})
respectively and $\xi \in I_-$ of (\ref{imi}). Then we have, if $n$ is large
enough%
\begin{equation}
d_{1j}(-\xi)\leq C^{\prime }n^{-1/2},\;d_{2j}(-\xi)\leq C^{\prime \prime }n^{-1/2}, \;\xi \in I_-,
\label{d1b}
\end{equation}%
where $C^{\prime }$ and $C^{\prime \prime }$ do not depend on $n$ and $j$.
\end{lemma}
%%%%%%%%%%%%%%%%%%%%%%%%%%%%%%%%%%%%%%%%%%%%%%%%%%%%%%%%%%%%%%%%%
\begin{proof}
We have by Schwarz inequality,
\begin{equation}
\hspace{-2cm}d_{1j}:=\mathbf{E}\{|a_{j}-\overline{h}_{n}|\}\leq \mathbf{E}%
^{1/2}\{|a_{j}-\overline{h}_{n}|^{2}\},  \label{d1d2}
\end{equation}%
%we consider $d_{2j}$ of (\ref{t211}). Use
and then the inequality %\begin{equation*}
$\left( a_{1}+a_{2}+a_{3}\right) ^{2}\leq 3(a_{1}^{2}+a_{2}^{2}+a_{3}^{2})$
%\end{equation*}%
yields
\begin{align}
& \hspace{-1.5cm}\mathbf{E}\{|a_{j}-\overline{h}_{n}|^{2}\}  \notag \\
& \hspace{-0.5cm} \leq 3\mathbf{E}\{|a_{j}-h_{jn}|^{2}\}+3\mathbf{E}\{|h_{jn}-h_{n}|^{2}\}+3%
\mathbf{E}\{|h_{n}-\overline{h}_{n}|^{2}\},  \label{3d}
\end{align}%
where $h_{jn}$ is defined in (\ref{hj}). It follows from
(\ref{Yj}) and (\ref{ggj}) that
\begin{equation}
a_{j}:=(G_{j}Y_{j},Y_{j})=n^{-1}(SG_{j}SX_{j},X_{j}).  \label{ajsx}
\end{equation}%
Denote by $\mathbf{E}_{j}\{\dots\}$ the (conditional) expectation with respect to
$X_{j}$ and $\mathbf{Var}_{j}\{\dots\}$ the corresponding variance (recall
that according to (\ref{Xn}) $\{X_{j}\}_{j=1}^{n}$ are the $n$-component i.i.d. vectors with i.i.d. components). Since $G_{j}$ is independent of $X_{j}$ by (\ref{cgj}),
we have from the above and (\ref{ea}) $\ \mathbf{E}_{j}\{a_{j}\}=h_{jn}$
(see (\ref{exfj})), thus the first term on the r.h.s. of (\ref{3d}) is%
\begin{equation}
\mathbf{E}\{|a_{j}-\mathbf{E}_{j}\{a_{j}\}|^{2}\}=\mathbf{E}\{\mathbf{E}%
_{j}\{|a_{j}-\mathbf{E}_{j}\{a_{j}\}|^{2}\}\}=:\mathbf{E}\{\mathbf{Var}%
_{j}\{a_{j}\}\}.  \label{varaj}
\end{equation}%
Next, (\ref{exa}), (\ref{nors}) and (\ref{rrho}) imply%
\begin{equation}
\mathbf{Var}_{j}\{a_{j}\}\leq \mu \ n^{-2}\mathrm{Tr}\,G_{j}R^{2}G_{j}\leq \mu
\rho ^{2}/n\xi ^{2},  \label{vaj}
\end{equation}%
since%
\begin{equation}
|\mathrm{Tr}\,A|\leq n||A||,  \label{trno}
\end{equation}%
and we obtain for the first term of (\ref{3d})%
\begin{equation}
\mathbf{E}\{|a_{j}-h_{jn}|^{2}\}\leq \mu \rho ^{2}/n\xi ^{2}.  \label{dj1}
\end{equation}%
Consider the second term of the r.h.s. of (\ref{3d}). Since $G(-\xi )$ and $%
G_{j}(-\xi )$ in the definitions (\ref{hns}) of $h_{n}$ and (\ref{hj}) of $%
h_{jn}$ are the resolvents of positive definite $\mathcal{M}_{n}$ and $%
\mathcal{M}_{n}-K_{j}L_{j}$, we use (\ref{troc}) with $A=\mathcal{M}_{n}$
and $C=R$ and (\ref{rrho}) to obtain %(writing $G $ and $G_{j} $ instead of
%and  $G_{j}(-\xi )$ here and below)
that $|h_{jn}-h_{n}|\leq \rho /n\xi $. Hence, we have for the second term of
the r.h.s. of (\ref{3d}) %
\begin{equation}
\mathbf{E}\{|h_{jn}-h_{n}|^{2}\}\leq \rho ^{2}/n^{2}\xi ^{2}.  \label{dj2}
\end{equation}%
As for the third term in the r.h.s. of (\ref{3d}), its bound follows from
Lemma \ref{l:mart} (ii) with $A=R$, yielding in view of (\ref{rrho})
\begin{equation*}
\mathbf{E}\{|h_{n}-\overline{h}_{n}|^{2}\}=\mathbf{Var}\{h_{n}\}\leq
C^{(2)}\rho ^{2}/n\xi ^{2}.
\end{equation*}%
Combining this bound with (\ref{dj1}) and (\ref{dj2}) and using then (\ref%
{d1d2}), we get the first bound in (\ref{d1b}).

To prove the second bound in (\ref{d1b}) we apply an analogous argument to
the r.h.s. of
\begin{align}
& \hspace{-1cm}d_{2j}=\mathbf{E}\{|b_{j}-\overline{c}_{n}|\}  \notag \\
& \leq \mathbf{E}\{|b_{j}-c_{jn}|\}+\mathbf{E}\{|c_{jn}-c_{n}|\}+\mathbf{E%
}\{|c_{n}-\overline{c}_{n}|\},  \label{d3j}
\end{align}%
where $c_{jn}=n^{-1}\mathrm{Tr}\,R^{2}\mathcal{G}G_{j}$ (cf. (\ref{hj})).

It follows from (\ref{ea}) and (\ref{bp}) that $\mathbf{E}%
_{j}\{b_{j}\}=c_{jn}$, hence, (cf. (\ref{varaj}))%
\begin{eqnarray*}
\mathbf{E}\{|b_{j}-c_{jn}|\} &=&\mathbf{E}\{|b_{j}-\mathbf{E}_{j}\{b_{j}\}|\}
\\
&=&\mathbf{E}\{\mathbf{E}_{j}\{|b_{j}-\mathbf{E}_{j}\{b_{j}\}|\}\}\leq
\mathbf{E}\{\mathbf{Var}_{j}^{1/2}\{b_{j}\}\}.
\end{eqnarray*}%
Using (\ref{exa}) with $C=R\mathcal{G}G_{j}$ and taking into account that $%
S^{2}=R$ and that $R$ and $\mathcal{G}$ commute (see (\ref{cg40})), we have
by (\ref{nors}), (\ref{norH}) and (\ref{trno})
\begin{equation*}
\mathbf{Var}_{j}\{b_{j}\}\leq \mu n^{-2}\mathrm{Tr\ }\mathcal{G}%
^{2}R^{3}G_{j}RG_{j}\leq \mu \rho ^{4}/n\xi ^{4},
\end{equation*}%
hence, the bound for the first term of (\ref{d3j})%
\begin{equation*}
\mathbf{E}\{|b_{j}-c_{jn}|\}\leq \mu ^{1/2}\rho ^{2}/n^{1/2}\xi ^{2}.
\end{equation*}%
Next, we have
\begin{equation*}
\mathbf{E}\{|c_{jn}-c_{n}|\}\leq \rho ^{2}/n\xi ^{2}
\end{equation*}%
(cf. (\ref{dj2})) for the second term of (\ref{d3j}) and
\begin{equation*}
\mathbf{E}\{|c_{n}-\overline{c}_{n}|\}\leq (C^{(2)})^{1/2}\rho
^{2}/n^{1/2}\xi ^{2}.
\end{equation*}%
by Lemma \ref{l:mart} with $A=R^{2}\mathcal{G}$ for the third term of (\ref%
{d3j}). Plugging the above three bound into (\ref{d3j}), we obtain the
second bound in (\ref{d1b}).
\end{proof}
The next lemma is a version of assertions given in Section 18.2 of
\cite{Pa-Sh:11}.
\begin{lemma}\label{l:mart}
Let $\mathcal{M}_{n}$ be given by (\ref{mncal}) in
which the entries of $X_{n}=\{X_{j\alpha }\}_{j,\alpha =1}^{n}$ of (\ref{Xn}) and the components of $%
b_{n}=\{b_{j }\}_{j=1}^{n}$ of (\ref{b}) are i.i.d. random variables.
Denote $\nu _{\mathcal{M}_{n}}$ the Normalized Counting Measure of $\mathcal{%
M}_{n}$ (see, e.g. (\ref{ncm})) and%
\begin{equation}
s_{n}(z)=n^{-1}\mathrm{Tr}\,AG(z),  \label{sn}
\end{equation}%
where $G(z)=(\mathcal{M}_{n}-z)^{-1}$ is the resolvent of $\mathcal{M}_{n}$
and $A$ is an $n\times n$ and $X_{n}$-independent matrix. We have:

\medskip
(i) for any $n$-independent interval $\Delta $ of spectral axis
\begin{equation}  \label{mnu4}
\mathbf{E}\{|\nu _{\mathcal{M}_{n}}(\Delta )-\mathbf{E}\{\nu _{\mathcal{M}%
_{n}}(\Delta )\}|^{4}\}\leq C^{(1)}/n^{2},
\end{equation}%
where $C^{(1)}$ is an absolute constant;

(ii) for any $n$-independent $\xi >0$%
\begin{equation*}
\mathbf{Var}\{s_{n}(-\xi )\}:=\mathbf{E}\{|s_{n}(-\xi )-\mathbf{E}%
\{s_{n}(-\xi )\}|^{2}\}\leq C^{(2)}||A||^{2}/n\xi ^{2},
\end{equation*}%
where $C^{(2)}$ is an absolute constant.
\end{lemma}
%%%%%%%%%%%%%%%%%%%%%%%%%%%%%%%%%%%%%%%%%%%%%%%%%%%%%%%%%%%%%%%
\begin{proof}
It follows from a general martingale difference argument (see \cite{Pa-Sh:11},
Proposition 18.1.1) that if $\psi :\mathbb{R}^{n^{2}}\rightarrow \mathbb{C}$,
%is a Borelian function, $
$\{X_{j}\}_{j=1}^{n}$ are i.i.d. random vectors, $%
\Psi =\psi (X_{1},\dots,X_{n})$, \ $\mathbf{E}_{j}\{\dots\}$ is the
expectation conditioned on $\{X_k\}_{k \neq j}$
%with respect to $X_{j}$
and%
\begin{equation*}
\Psi _{j}=\mathbf{E}_{j+1}\dots \mathbf{E}_{n}\{\Psi \},
\end{equation*}%
then
\begin{equation}
\mathbf{E\{|}\Psi -\mathbf{E\{}\Psi \mathbf{\}|}^{2p}\mathbf{\}\leq }%
C_{p}n^{p-1}\sum_{j=1}^{n}\mathbf{E}\{|\Psi _{j}-\mathbf{E}_{j}\{\Psi
_{j}\}|^{2p}\},  \label{mart}
\end{equation}%
where $C_{p}$ depends only on $p$.

Choose $\Psi =\nu _{\mathcal{M}_{n}}(\Delta )$ and the rows $\{X_{j\alpha
}\}_{\alpha =1}^{n}$ of $X_{n}$ of (\ref{Xn}) as $X_{j}$ and write $\nu _{%
\mathcal{M}_{n}}(\Delta )=\nu _{\mathcal{M}_{n}^{(j)}}(\Delta )+\mu
_{jn}(\Delta )$, where $\mathcal{M}_{n}^{(j)}$ is defined in (\ref{mmj}).
Since $\mathcal{M}_{n}-\mathcal{M}_{n}^{(j)}=K_{j}L_{j}$ is a rank-one
matrix, we can use the interlacing property of eigenvalues of a hermitian
matrix and its rank-one perturbation (see \cite{Ho-Jo:13}, Section 4.3
and formula (\ref{r12}) of this paper) to show
that $|\mu _{jn}(\Delta )|\leq 1/n$ for any $\Delta \in \mathbb{R}_{+}$
and any realization of random parameters. Hence, taking into account that
$%
\mathcal{M}_{n}^{(j)}$ does not depend on $X_{j}$, we obtain%
\begin{equation}
|\Psi _{j}-\mathbf{E}_{j}\{\Psi _{j}\}|=|\mu _{jn}(\Delta )-\mathbf{E}%
_{j}\{\mu _{jn}(\Delta )\}|\leq 2/n.  \label{cpj}
\end{equation}%
This and (\ref{mart}) with $p=2$ imply assertion (i) of the lemma with $%
C^{(1)}=2^{4}C_{2}$.

To prove assertion (ii) we choose $p=1$ in (\ref{mart})  ,
\begin{equation*}
\Psi =s_{n}(-\xi )=n^{-1}\mathrm{Tr}\,AG(-\xi )
\end{equation*}%
and the same $X_{j}$'s. If $G_{j}$ is given by (\ref{cgj}), then we have by (%
\ref{troc}) with $A$ and $B$ as in (\ref{mmj}) and $C=A:$
\begin{equation*}
n^{-1}\mathrm{Tr}\,AG=n^{-1}\mathrm{Tr}\,AG_{j}-l_{jn},\;\;|l_{jn}|\leq
n^{-1}\xi ^{-1}||A||.
\end{equation*}%
Hence, in this case (cf. (\ref{cpj}))
\begin{eqnarray*}
\mathbf{E}\{|\Psi _{j}-\mathbf{E}_{j}\{\Psi _{j}\}|^{2}\} &=&\mathbf{E}\{%
\mathbf{E}_{j}\{|\Psi _{j}-\mathbf{E}_{j}\{\Psi _{j}\}|^{2}\}\} \\
&=&\mathbf{E}\{\mathbf{E}_{j}\{|l_{jn}-\mathbf{E}_{j}\{l_{jn}\}|^{2}\}\}\leq
\mathbf{E}\{\mathbf{E}_{j}\{|l_{jn}|^{2}\}\}.
\end{eqnarray*}%
%We have from (\ref{ljn}), (\ref{aj}) and (\ref{Kj})
%\begin{equation*}
%|l_{jn}|=n^{-1}(G_{j}Y_{j},Y_{j})^{-1}(G_{j}AG_{j}Y_{j},Y_{j})\leq
%n^{-1}||A||(G_{j}Y_{j},Y_{j})^{-1}(G_{j}^{2}Y_{j},Y_{j}),
%\end{equation*}%
%and then (\ref{GG2}) implies that
%$|l_{jn}|\leq n^{-1}\xi ^{-1}||A||$. We
and we obtain the bound%
\begin{equation*}
\mathbf{E}\{|\Psi _{j}-\mathbf{E}_{j}\{\Psi _{j}\}|^{2}\}\leq n^{-2}\xi
^{-2}||A||^{2}
\end{equation*}%
implying assertion (ii) of the lemma.
\end{proof}
%\begin{remark}
%\label{r:mxind} (i) The independence of random vectors $\mathcal{Y}_j$ in (%
%\ref{cmyy}) is the main reason to pass from the matrices $M_n^l$ given by (%
%\ref{m1}) and (\ref{m2m1}) to the matrices $\mathcal{M}_n$ given by (\ref%
%{cm1}), (\ref{cm1m2}) and (\ref{mncal}).
%\end{remark}
%%%%%%%%%%%%%%%%%%%%%%%%%%%%%%%%%l:xlyl%%%%%%%%%%%%%%%%%%%%%%%%%%%%
\begin{lemma}
\label{l:clt} Let $\{X_{\alpha }\}_{\alpha =1}^{n}$ be i.i.d. random
variables satisfying (cf. (\ref{wga}) and (\ref{Xn}))
\begin{equation}
\mathbf{E}\{X_{\alpha }\}=0,\ \mathbf{E}\{X_{\alpha }^{2}\}=1,\ \mathbf{E}%
\{X_{\alpha }^{4}\}=m_{4}<\infty   \label{xcond}
\end{equation}%
and $\{x_{\alpha n}\}_{\alpha =1}^{n}$ be collection of real numbers
satisfying (\ref{qqn}) and (\ref{xbou}).

Then the random variable (cf. (\ref{nukn0}))
\begin{equation*}
\eta _{n}=n^{-1/2}\sum_{\alpha =1}X_{\alpha }x_{\alpha n}
\end{equation*}%
converges in distribution to
\begin{equation*}
(q-\sigma _{b}^{2})^{1/2}\ \gamma,
\end{equation*}%
where $q$ and $\sigma _{b}^{2}$ are given by (\ref{qqn}) and (\ref{b}) and $%
\gamma $ is the standard Gaussian random variable.
\end{lemma}
%%%%%%%%%%%%%%%%%%%%%%%%%%%%%%%%%%%%%%%%%%%%%%%%%%%%%%%%%%%%%%%
\begin{proof}
We will use the Central Limit Theorem for independent and not necessarily
identically distributed random variables $\{\xi _{\alpha n}\}_{\alpha =1}^{n}
$ with
\begin{align*}
& \mathbf{E}\{\xi _{\alpha n}\}=0,\;\mathbf{E}\{\xi^2 _{\alpha n}\}=\sigma
_{\alpha n}^{2},\;\Xi _{n}=\sum_{\alpha =1}^{n}\xi _{\alpha n}, \\
& \hspace{0.7cm}\Sigma _{n}^{2}:=\mathbf{Var}\{\Xi _{n}\}=\sum_{\alpha
=1}^{n}\sigma _{\alpha n}^{2}.
\end{align*}%
In this case $\Sigma _{n}^{-1}\Xi _{n}$ converges in distribution to \ the
standard Gaussian variable $\gamma $ if for any $\tau >0$%
\begin{equation}
\lim_{n\rightarrow \infty }\Sigma _{n}^{-2}\sum_{\alpha =1}^{n}\mathbf{E}%
\{\xi _{\alpha n}^{2}I(|\xi _{\alpha n}|-\tau \Sigma _{n})\}=0,  \label{Lin}
\end{equation}%
where $I$ is the indicator of $\mathbb{R}_{+}$ (see \cite{Sh:96}, Section
III.4). Choosing $X_{\alpha }x_{\alpha n}$ as $\xi _{\alpha n}$, it is easy
to find from (\ref{qqn}) that (\ref{Lin}) is equivalent to
\begin{equation}
\lim_{n\rightarrow \infty }n^{-1}\sum_{\alpha =1}^{{}}x_{\alpha n}^{2}%
\mathbf{E}\{X_{\alpha }^{2}I(|X_{\alpha }|-\tau \sqrt{n}/x_{\alpha n})\}=0.
\label{Lin1}
\end{equation}%
It follows from (\ref{xcond}) that $\mathbf{E}\{X_{\alpha }^{2}I(|X_{\alpha
}|-\tau \sqrt{n}/x_{\alpha n})\}\leq m_{4}x_{\alpha n}^{2}/\tau ^{2}n$ and
then the l.h.s. of (\ref{Lin1}) is bounded by%
\begin{equation*}
\lim_{n\rightarrow \infty }(n\tau )^{-2}\sum_{\alpha =1}^{{}}x_{\alpha
n}^{4},\ \tau >0,
\end{equation*}%
which is zero in view of (\ref{xbou}).

Likewise, $\Sigma _{n}^{-1}\Xi _{n}$
is equivalent to $(q_{n}-\sigma _{b}^{2})^{-1/2}\eta _{n}$, hence,
converges in distribution to $\gamma $.
\end{proof}

\smallskip
The next lemma deals with asymptotic properties of the vectors of
activations $x^{l}$ in the $l$th layer, see (\ref{rec}). It is an extended
version (treating the convergence with probability 1) of assertions proved
in \cite{Ma-Co:16,Po-Co:16,Sc-Co:17}.
%%%%%%%%%%%%%%%%%%%%%%%%%%%%%%%%%%%%%%%%%%%%%%%%%%%%%%%%%%%%%%%
\begin{lemma}
\label{l:xlyl} Let $y^{l}=\{y_{j}^{l}\}_{j=1}^{n},\;l=1,2,\dots$ be
post-affine random vectors defined in (\ref{rec}) -- (\ref{wga}) with $x^{0}$
satisfying (\ref{q0}), $\chi :\mathbb{R}\rightarrow \mathbb{R}$ be a bounded
continuous function and $\Omega _{l}$ be defined in (\ref{oml}). Set
\begin{equation}
\chi _{n}^{l}=n^{-1}\sum_{j_{l}=1}^{n}\chi (y_{j_{l}}^{l}),\;l\geq 1.
\label{chiln}
\end{equation}%
Then there exists $\overline{\Omega }_{l}\subset \Omega _{l},\;\mathbf{P}(%
\overline{\Omega }_{l})=1$ such that for every $\omega _{l}\in \overline{%
\Omega }_{l}$ (i.e., with probability 1) the limits%
\begin{equation}
\chi ^{l}:=\lim_{n\rightarrow \infty }\chi _{n}^{l},\;l=1,2, \dots  \label{lql}
\end{equation}%
exist, are not random (do not depend on the realizations with probability 1) and given by the formula
\begin{equation}
\chi ^{l}=\int_{-\infty }^{\infty }\chi \big(\gamma (q^{l}-\sigma _{b}^{2})^{1/2}%
+b\big)\Gamma (d\gamma )F(db),\;l=1,2, \dots  \label{lqg}
\end{equation}%
valid on $\overline{\Omega }_{l}$ with\ $\Gamma (d\gamma )=(2\pi
)^{-1/2}e^{-\gamma ^{2}/2}d\gamma $ being the standard Gaussian probability
distribution, $F$ is the common probability law of $\{b_{j_l}^{l}\}_{j_l=1}^n$ in (\ref{bga}) and $%
q^{l}$ defined recursively by the formula
\begin{equation}
q^{l}=\int_{-\infty }^{\infty }\varphi ^{2}\big(\gamma (q^{l-1}-\sigma_b^2)^{1/2}+b\big)\Gamma
(d\gamma )F(db)+\sigma _{b}^{2},\;l=2,3, \dots  \label{ql}
\end{equation}%
with $q^{1}$ given in (\ref{q0}).

In particular, we have with probability 1 formula (\ref{nukal}) for the weak
limit $\nu _{K^{l}}$ of the Normalized Counting Measure $\nu _{K_{n}^{l}}$
of diagonal random matrix $K_{n}^{l}$ of (\ref{kan}).
\end{lemma}

%%%%%%%%%%%%%%%%%%%%%%%%%%%%%%%%%%%%%%%%%%%%%%%%%%%%%%%%%%%%%%%
\begin{proof}
Set $l=1$ in (\ref{chiln}) Since $\{b_{j_{1}}^{1}\}_{j_{1}=1}^{n}$ and $%
\{X_{j_{1},j_{0}}^{1}\}_{j_{1},j_{0}=1}^{n}$ are i.i.d. random variables satisfying (\ref{bga}) -- (\ref{wga}), it follows from (\ref{rec}) that the components $%
\{y_{j_{1}}^{1}\}_{j_{1}=1}^{n}$ of $y^{1}$ are also i.i.d. random variables
of zero mean and variance $q_{n}^{1}$ of (\ref{q0}). Since $\chi $ is
bounded, the collection $\{\chi (y_{j_{1}}^{1})\}_{j_{1}=1}^{n}$ consists of
bounded i.i.d random variables defined for all $n$ on the same probability
space $\Omega _{1}$ generated by (\ref{xinf}) and (\ref{binf}) with $l=1$.
This allows us to apply to $\{\chi (y_{j_{1}}^{1})\}_{j_{1}}^{n}$ the strong
Law of Large Numbers implying (\ref{lql}) for $l=1$ together with the formula%
\begin{equation}
\chi ^{1}=\lim_{n\rightarrow \infty }\mathbf{E}\{\chi (y_{1}^{1})\}
\label{chi0}
\end{equation}%
valid on a certain $\overline{\Omega }_{1}\subset \Omega _{1}=\Omega ^{1}$, $%
\mathbf{P}(\overline{\Omega} _{1})=1$, see (\ref{oml}).

To get (\ref{lqg}) for $l=1$ recall that according to (\ref{rec}) and (\ref%
{bga}) -- (\ref{wga})%
\begin{equation*}
y_{1}^{1}=\eta _{1}^{1}+b_{1}^{1},\;\eta
_{1}^{1}=n^{-1/2}\sum_{j_{0}=1}^{n}X_{1j_{0}}^{1}x_{j_{0}}^{0}
\end{equation*}%
and $\eta _{1}^{1}$ and $b_{1}^{1}$ are independent. Hence,%
\begin{equation*}
\mathbf{E}\{\chi (y_{1}^{1})\}=\int \chi (\eta +b)\Gamma_{n}^{}(d\eta )F(db),
\end{equation*}%
where $\Gamma_{n}$ is the probability law of $\eta _{1}^{1}$. Passing here to
the limit $n\rightarrow \infty $ and using Lemma \ref{l:clt}, and (\ref{q0})
-- (\ref{qLin}), we obtain (\ref{lqg}) for $l=1$.

Consider now the case $l=2$. Since $\{X^{1},b^{1}\}$ and $\{X^{2},b^{2}\}$
are independent, we can fix $\omega _{1}\in \overline{\Omega }_{1},\ \mathbf{%
P}(\overline{\Omega }_{1})=1$ (a realization of $\{X^{1},b^{1}\}$) and apply
to $\chi _{n}^{2}$ of (\ref{chiln}) the same argument as that for the case $%
l=1$ above to prove that for every $\omega _{1}\in \overline{\Omega }_{1}$
there exists $\overline{\Omega }^{2}(\omega ^{1})\subset \Omega ^{2},\;%
\mathbf{P}(\overline{\Omega }^{2}(\omega ^{1}))=1\,$\ on which we have the
analog of (\ref{chi0})
\begin{equation}
\chi ^{2}(\omega ^{1},\omega ^{2})=\lim_{n\rightarrow \infty }\mathbf{E}%
_{\{X^{2},b^{2}\}}\{\chi (y_{1}^{2})\}  \label{q2oo}
\end{equation}%
%valid for every pair $(\omega ^{1},\omega ^{2}),\;\omega ^{1}\in %\overline{%
%\Omega _{1}}$, $\omega _{2}\in \overline{\Omega _{2}}(\omega %_{1})$ with%
where $\mathbf{E}_{\{X^{2},b^{2}\}}\{...\}$ denotes the expectation with
respect to $\{X^{2},b^{2}\}$ only. Now, by using again Lemma \ref{l:clt} and
the Fubini theorem we obtain that there exists $\overline{\Omega }%
_{2}\subset \Omega _{2}=\Omega ^{1}\otimes \Omega ^{2},\;$ $\mathbf{P}(%
\overline{\Omega }_{2})=1$ on which we have (\ref{lql}) for $l=2$ with
\begin{equation*}
q^{2}=\lim_{n\rightarrow \infty
}n^{-1}\sum_{j_{1}=1}^{n}(x_{j_{1}}^{1})^{2}+\sigma _{b}=\lim_{n\rightarrow
\infty }n^{-1}\sum_{j_{1}=1}^{n}(\varphi (y_{j_{1}}^{1}))^{2}+\sigma _{b}.
\end{equation*}%
The limit in the r.h.s. above exists with probability 1 on $\Omega ^{1}$ and
equals the r.h.s. of (\ref{ql}) for $l=2$ just because it is a particular
case of (\ref{lqg}) for $\chi =\varphi ^{2}$ and $l=2$.

This proves the validity (\ref{lql}) -- (\ref{ql}) for $l=2$ with
probability 1. Analogous argument applies for $l=3,4,\dots$.

To prove (\ref{nukal}) it suffices to prove the validity with probability 1
of
\begin{equation*}
\lim_{n\rightarrow \infty }\int_{-\infty }^{\infty }\psi (\lambda )\nu
_{K_{n}^{l}}(d\lambda )=\int_{-\infty }^{\infty }\psi (\lambda )\nu
_{K^{l}}(d\lambda )
\end{equation*}%
for any bounded and continuous $\psi :\mathbb{R\rightarrow R}$.

In view of (\ref{rec}), (\ref{D}) and (\ref{kan}) the relation can be
written as%
\begin{equation*}
\lim_{n\rightarrow \infty }n^{-1}\sum_{j_{l}=1}^{n}\psi \Big(\big(\varphi ^{\prime
}(y_{j_{l}}^{l})\big)^{2}\Big)=\int_{-\infty }^{\infty }\psi \Big(\big(\varphi ^{^{\prime
}}(\gamma (q^{l}-\sigma_b^2)^{1/2}+b)\big)^{2}\Big)\Gamma (d\gamma )F(db),\;l\geq 1.
\end{equation*}%
The l.h.s. here is a particular case of (\ref{chiln}) -- (\ref{lql}) for $%
\chi =\psi \circ \varphi ^{\prime 2}$, thus, it equals the r.h.s. of (\ref%
{lqg}) for this $\chi $.
\end{proof}
%%%%%%%%%%%%%%%%%%%%%%%%%%%%%%%%%%%%%%%%%%%%%%%%%%%%%%%%%%%%%%

\medskip
The next lemma provides the unique solvability of the system (\ref{hk}) -- (%
\ref{kh}). The lemma is a streamlined version of Lemma 3.12 in \cite{Pa:20}.
Note that in the course of proving Theorem \ref{t:ind} it was
found that the system has at least one solution $(h,k)$ analytic in $\mathbb{C} \setminus \mathbb{R}_+$ and such that $h$ satisfies the $n \to \infty$ versions of (\ref{hrep}) -- (\ref{pos}). This is used below to determine the class of functions in which the unique solvability holds.
%%%%%%%%%%%%%%%%%%%%%%%%%%%%%%%%%%%%%%%%%%%%%%%%%%%%%%%%%%%%%

\begin{lemma}
\bigskip \label{l:uniq} The system (\ref{hk}) -- (\ref{kh}) with $\nu _{R}$
and $\ \nu _{K}$ satisfying
\begin{equation}
\nu _{K}(\mathbb{R}_{+})=1,\;\nu _{R}(\mathbb{R}_{+})=1  \label{mkr}
\end{equation}%
$\;$ and (cf. ((\ref{r2}))%
\begin{equation}
\int_{0}^{\infty }\lambda ^{2}\nu _{K}(d\lambda )=\kappa _{2}<\infty
,\;\int_{0}^{\infty }\lambda ^{2}\nu _{R}(d\lambda )=\rho _{2}<\infty
\label{nukr}
\end{equation}%
has a unique solution in the class of pairs $(h,k)$ of functions  defined in $%
\mathbb{C}\setminus \mathbb{R}_{+}$ and such that $h$ is analytic in $%
\mathbb{C}\setminus \mathbb{R}_{+}$, continuous and positive on the open
negative semi-axis and satisfies (\ref{hcond}).
%with $r_{2}$ replaced by $ \rho _{2}$ of (\ref{nukr}).

In addition

(i) the function $k$ is analytic in $\mathbb{C}\setminus \mathbb{R}_{+}$,
continuous and positive on the open negative semi-axis and (cf. (\ref{hcond}%
))%
\begin{equation}
\Im k(z)\Im z<0\;\mathrm{for\;}\Im z\neq 0,\;0<k(-\xi )\leq \kappa
_{2}^{1/2}\;\mathrm{for\;}\xi >0  \label{imk}
\end{equation}%
with $\kappa _{2}$ of (\ref{nukr});

(ii) if the measures $\nu _{K^{(p)}}$ and $\nu _{R^{(p)}}, \; p=1,2, \dots$ have uniformly
in $p$ bounded second moments (see (\ref{nukr})) and converge weakly to $%
\nu _{K}$ and $\nu _{R}$ also satisfying (\ref{nukr}), then the sequences of the
corresponding solutions $\{h^{(p)},k^{(p)}\}_{p}$ of the system (\ref{hk})
-- (\ref{kh}) converges pointwise in $\ \mathbb{C}\setminus \mathbb{R}_{+}$
to the solution $(h,k)$ of the system corresponding to the limiting measures
$(\nu _{K},\nu _{R})$.
\end{lemma}
\begin{proof}
Note that in the course of proving Theorem \ref{t:ind} it was
proved that the system has at least one solution satisfying the conditions of the lemma.

Let us prove assertion (i) of the lemma. It follows from (\ref{kh}), (%
\ref{nukr}) and the analyticity of $h$ in $\mathbb{C}\setminus \mathbb{R}%
_{+} $ that $k$ is also analytic in $\mathbb{C}\setminus \mathbb{R}_{+}$.
Next, for any solution of (\ref{hk}) -- (\ref{kh}) we have \ from (\ref{kh})
with $\Im z\neq 0$%
\begin{equation}
\Im k(z)=-\Im h(z)\int_{0}^{\infty }\frac{\lambda ^{2}\nu _{K}(d\lambda )}{%
|h(z)\lambda +1|^{2}}  \label{imkh}
\end{equation}%
and then (\ref{hcond}) yields (\ref{imk}) for $\Im z\neq 0$, while (\ref{kh})
with $z=-\xi <0$, i.e.,
\begin{equation*}
k(-\xi )=\int_{0}^{\infty }\frac{\lambda \nu _{K}(d\lambda )}{h(-\xi
)\lambda +1},
\end{equation*}%
the positivity of $h(-\xi )$ (see (\ref{hcond})),  (\ref{mkr}) and Schwarz
inequality yield (\ref{imk}) for $z=-\xi $.

Let us prove now that the system (\ref{hk}) -- (\ref{kh}) is uniquely
solvable in the class of pairs of functions $(h,k)$ analytic in $\mathbb{C}%
\setminus \mathbb{R}_{+}$ and satisfying (\ref{hcond}) and (\ref{imk}).
The argument below is a version of that used in \cite{Pa-Sh:11}, Lemma 2.2.6
for the deformed Wigner ensemble.

Assume that there exist two different solutions $(h_{1},k_{1})$ and $%
(h_{2},k_{2})$ of (\ref{hk}) -- (\ref{kh}), i.e., there exists $z_{0} \in \mathbb{C}\setminus \mathbb{R}_{+}$, where at least one of two functions
$\delta h=h_{1}-h_{2}, \; \delta k=k_{1}-k_{2}$ is not zero. Since $\delta h$ and $\delta k$ are
analytic in $\mathbb{C}\setminus \mathbb{R}_{+}$, we can assume without loss
of generality that $\Im z_{0} \neq 0$. It follows then from (\ref{hk}%
) -- (\ref{kh}) that
\begin{equation}
\delta h(z_{0})=-\delta k(z_{0})\;I_{k}(z_{0}),\;\delta k(z_{0})=-\delta
h(z_{0})\;I_{h}(z_{0}),  \label{dhdk}
\end{equation}%
where%
\begin{eqnarray*}
I_{k}(z) &=&\int_{0}^{\infty }\frac{\lambda ^{2}\nu _{R}(d\lambda )}{%
(\lambda k_{1}(z)-z)(\lambda k_{2}(z)-z)},\; \\
I_{h}(z) &=&\int_{0}^{\infty }\frac{\lambda ^{2}\nu _{K}(d\lambda )}{%
(\lambda h_{1}(z)+1)(\lambda h_{2}(z)+1)}.
\end{eqnarray*}%
Viewing (\ref{dhdk}) as a system of linear equations for $\delta h$ and $%
\delta k$ that has a non-trivial solution, we conclude that
\begin{equation}
1=I_{k}(z_{0})I_{h}(z_{0}),\; \Im z_{0} \neq 0.  \label{det1}
\end{equation}%
On the other hand, we have by Schwarz inequality%
\begin{equation}
|I_{k}(z)|\leq (A_{k_{1}}(z)A_{k_{2}}(z))^{1/2},\;|I_{h}(z)|\leq
(A_{h_{1}}(z)A_{h_{2}}(z))^{1/2},  \label{iaia}
\end{equation}%
where%
\[
A_{k}(z)=\int_{0}^{\infty }\frac{\lambda ^{2}\nu _{R}(d\lambda )}{|\lambda
k(z)-z|^{2}},\;A_{h}(z)=\int_{0}^{\infty }\frac{\lambda ^{2}\nu
_{K}(d\lambda )}{|\lambda h(z)+1|^{2}}.
\]

In addition, the imaginary parts of (\ref{hk}) and (\ref{kh}) yield%
\begin{eqnarray*}
\Im k(z) &=&-\Im h(z)A_{h}(z), \\
\Im h(z) &=&-\Im k(z)A_{k}(z)+\Im z\int_{0}^{\infty }\frac{%
\lambda \nu _{R}(d\lambda )}{|\lambda k(z)-z|^{2}},
\end{eqnarray*}%
hence, by (\ref{hcond}),%
\[
0<A_{k}(z)A_{h}(z)=1-\frac{\Im z}{\Im h(z)}\int_{0}^{\infty }\frac{%
\lambda \nu _{R}(d\lambda )}{|\lambda k(z)-z|^{2}}<1,\;\Im z\neq 0.
\]%
This and (\ref{iaia}) lead to the strict inequality
\[
|I_{k}(z)I_{h}(z)|^{2}\leq
(A_{k_{1}}(z)A_{h_{1}}(z))(A_{k_{2}}(z)A_{h_{2}}(z))<1,\;\;\Im z\neq 0
\]%
which contradicts (\ref{det1}).

Let us prove assertion (ii) of the lemma. Since $h^{(p)}$ and $k^{(p)}$ are
analytic and uniformly in $p$ bounded outside the closed positive semiaxis,
there exist subsequences $\{h^{(p_{j})},k_{{}}^{(p_{j})}\}_{j}$ converging
pointwise in $\mathbb{C}\setminus \mathbb{R}_{+}$ to a certain analytic pair
$(\widetilde{h},\widetilde{k})$. Let us show that $(\widetilde{h},\widetilde{%
k})=(h,k)$. It suffices to consider real negative $z=-\xi >0$ (see (\ref%
{imi})). Write for the analog of (\ref{kh}) for $\nu _{K^{(p)}}$:%
\begin{eqnarray*}
k^{(p)} &=&\int_{0}^{\infty }\frac{\lambda \nu _{K^{(p)}}(d\lambda )}{%
h^{(p)}\lambda +1} \\
&=&\int_{0}^{\infty }\frac{\lambda \nu _{K^{(p)}}(d\lambda )}{\widetilde{h}%
\lambda +1}+(\widetilde{h}-h^{(p)})\int_{0}^{\infty }\frac{\lambda ^{2}\nu
_{K^{(p)}}(d\lambda )}{(h^{(p)}\lambda +1)(\widetilde{h}\lambda +1)}.
\end{eqnarray*}%
Putting here $p=p_{j}\rightarrow \infty $, we see that the l.h.s. converges
to $\widetilde{k}$, the first integral on the right converges to the r.h.s
of (\ref{kh}) with $\widetilde{h}$ instead of $h$ since $\nu _{K^{(p)}}$
converges weakly to $\nu _{K}$, the integrand is bounded and continuous
and the second integral is bounded in $p$ since $h^{(p)}(-\xi )>0$, $%
\widetilde{h}(-\xi )>0$ and the second moment of $\nu _{K^{(p)}}$ is bounded
in $p$ according to (\ref{nukr}), hence, the second term vanishes as $%
p=p_{j}\rightarrow \infty $. An analogous argument applied to (\ref{hk}%
) show $(\widetilde{h},\widetilde{k})$ is a solution of (\ref{hk}) -- (\ref%
{kh}) and then the unique solvability of the system implies that $(%
\widetilde{h},\widetilde{k})=(h,k)$.
\end{proof}
\textbf{Acknowledgment}. We
are grateful to the anonymous referee for the careful reading of our manuscript
and suggestions which helped us to improve considerably the presentation.
L.P. is grateful to the Ecole Normale Sup\'eriore (Paris) for the invitation
to the Department of Physics, where the final version of the paper was prepared.

\end{document}